\documentclass[12pt]{amsart}

\usepackage{amscd,amssymb,amsmath,latexsym}
\usepackage[mathscr]{euscript}
\usepackage{epsfig}

\newcommand{\N}{\mathbb N}

\newcommand{\cB}{\mathcal B}

\newtheorem{Thm}{Theorem}[section]
\newtheorem{Cor}[Thm]{Corollary}

\newtheorem{Prop}[Thm]{Proposition}
\newtheorem{Rmk}[Thm]{Remark}
\newtheorem{Ex}[Thm]{Example}
\newtheorem{Def}[Thm]{Definition}

\begin{document}
\title{Quantum Kac's chaos}

\author{George Androulakis}
\address{Department of Mathematics, University of South Carolina, 
Columbia, SC 29208}
\email{giorgis@math.sc.edu}

\author{Rade Musulin}
\address{Department of Mathematics, University of South Carolina, 
Columbia, SC 29208}
\email{musulin@math.sc.edu}



\thanks{The article is part of the second author's Ph.D. thesis which is prepared at the University of South Carolina under the supervision of the first author.}

\maketitle

\begin{abstract}
We study the notion of quantum Kac's chaos which was implicitly introduced by Spohn and explicitly formulated by Gottlieb. We prove the analogue of a result of Sznitman which gives the equivalence of Kac's chaos to $2$-chaoticity and to convergence of empirical measures. Finally we give a simple, different proof of a result of Spohn which states that chaos propagates with respect to certain Hamiltonians that define the evolution of the mean field limit for interacting quantum systems.
\end{abstract}

\section{The Motivation Behind Kac's Chaos}

The origins of chaos, as discussed in this paper, dates back to Kac. In 1956, Kac \cite{Kac} was interested in solving the non-linear integro-differential equation known as the Boltzmann equation \cite[Equation~(1.1)]{Kac}. The solution to the Boltzmann equation is a family $(f^{(N)})_{N=1}^\infty$ of probability density functions, where $f^{(N)}$ describes the velocities and positions of $N$ dilute gas molecules moving in $\mathbb{R}^3$, interacting via elastic binary collisions. The non-linearity of the Boltzmann equation provided difficulty in obtaining the existence of its solution. 

If the gas is restricted to a container of fixed volume, there are no external forces, and the number $N$ of molecules is assumed to be equidistributed, then $f^{(N)}$ depends on the velocities of the $N$ gas molecules and time, thus having $3N+1$ real variables.  Then the Boltzmann equation takes a simplified reduced form \cite[Equation~(1.3)]{Kac} which is still a non-linear integro-differential equation. Further assuming that the kinetic energy of the system remains constant proportional to $N$, the $3N$ variables representing velocity lie on a sphere of radius $\sqrt{N}$ in $\mathbb{R}^{3N}$, and in order to obtain a further simplified version of the Bolzmann equation, one can replace the $3N$ real variables by one real variable $x$. This further reduces the Boltzmann equation to the reduced Boltzmann equation
\cite[Equation~(3.5)]{Kac}:
\begin{equation}\label{reducedBoltzmanneq}
\frac{\partial f (x,t)}{\partial t} = \frac{\nu}{2 \pi} \int_{-\infty}^\infty \int_0 ^{2\pi} \{ f(x \cos \theta + y \sin \theta , t)
f(-x \sin \theta + y \cos \theta , t) - f(x,t)f(y,t) \} d \theta d y.
\end{equation} 

Kac further introduced a linear differential equation which he called the ``Master Equation'' \cite[Equation (2.6)]{Kac}. If $\phi^{(N)}$ is a solution to Kac's ``Master Equation'', $\phi_1^{(N)}$ and $\phi_2^{(N)}$ will denote the first and second marginals of $\phi^{(N)}$, respectively, i.e.
$$
\phi^{(N)}_1(x,t)= \int_{x_2^2+ \ldots + x_N^2 = N-x^2} \phi^{(N)}(x,x_2, \ldots , x_N, t) d\sigma_1 (x_2, \ldots , x_N)
$$
and
$$
\phi^{(N)}_2(x,y,t)= \int_{x_3^2+ \ldots + x_N^2 = N-x^2-y^2} \phi^{(N)}(x,y,x_3, \ldots , x_N, t) d\sigma_2 (x_2, \ldots , x_N)
$$ where $\sigma_1,\sigma_2$ are normalized uniform measures on the spheres of $\mathbb{R}^{N-1}$ and $\mathbb{R}^{N-2}$ respectively, centered at the origin and having radii $\sqrt{N-x^2}$ and $\sqrt{N-x^2-y^2}$ respectively. Kac \cite{Kac} noticed that if $\lim\limits_{N\rightarrow\infty}\phi_1^{(N)}(x,0)$ exists weakly in $L^1(\mathbb{R})$ and $\lim\limits_{N \rightarrow \infty}\phi_2^{(N)}(x,0)$ exists weakly in $L^1(\mathbb{R}^2)$, and 
\begin{eqnarray}\label{2cp10}
\lim\limits_{N \rightarrow\infty}\phi_2^{(N)}(x,y,0)=\lim\limits_{N \rightarrow\infty}\phi_1^{(N)}(x,0)\lim\limits_{N \rightarrow\infty}\phi_1^{(N)}(y,0),
\end{eqnarray} then the same limits exist at any later time $t$, and satisfy
\begin{eqnarray}\label{2cp1}
\lim\limits_{N \rightarrow\infty}\phi_2^{(N)}(x,y,t)=\lim\limits_{N \rightarrow\infty}\phi_1^{(N)}(x,t)\lim\limits_{N \rightarrow\infty}\phi_1^{(N)}(y,t).
\end{eqnarray} Then equation (\ref{2cp1}) implies that the function $f$ defined by $$f(x,t):=\lim\limits_{N \rightarrow \infty}\phi_1^{(N)}(x,t)$$ satisfies equation (\ref{reducedBoltzmanneq}). Hence Kac proved the existence of the solution to the reduced Boltzmann equation for $N=1$. Kac \cite{Kac} referred to the property in equation (\ref{2cp1}) for a fixed $t\geq 0$ as the ``Boltzmann property''. Whenever equation (\ref{2cp10}) implies equation (\ref{2cp1}) for all times $t>0$, we say that the ``Boltzmann property propagates in time''. Hence Kac \cite{Kac} proved that the Boltzmann property propagates in time for his ``Master Equation''.

Many authors including McKean \cite{McKean}, Johnson \cite{Johnson}, Tanaka \cite{Tanaka}, Ueno \cite{Ueno}, Gr\"{u}nbaum \cite{Grunbaum},      Graham and M\'{e}l\'{e}ard \cite{GrahamMeleard},
 Sznitman \cite{Sznitman}, Mischler \cite{Mischler}, Carlen, Carvalho and Loss \cite{CarlenCarvalhoLoss},
 Mischler and Mouhot \cite{MischlerMouhot} have abstracted the idea of the ``Boltzmann property'' to a sequence of probability measures on a topological space. Instead of having the ``Boltzmann property'', the sequence of probability measures nowadays are said to be chaotic. In order to discuss chaotic sequences of probability measures, these authors first define the notion of a symmetric probability measure.

\begin{Def} \label{Def:symmetricClassical}
Let $E$ be a topological space, $N$ be a positive integer, $\mu_N$ be a probability measure on the Borel subsets of $E^N$.
Then $\mu_N$ is called \textbf{symmetric} if for any $N$-many continuous scalar-valued bounded functions on $E$, $\phi_1,\phi_2,...,\phi_N$,
\begin{eqnarray*}
\int_{E^N} \phi_1(x_1)\phi_2(x_2) \cdots \phi_N(x_N) d \mu_N =\int_{E^N} \phi_1(x_{\pi(1)})\phi_2(x_{\pi(2)}) \cdots \phi_n (x_{\pi(N)}) d \mu_N
\end{eqnarray*}
for any permutation $\pi$ of $\{ 1, \ldots , N \}$.
\end{Def} A chaotic sequence of probability measures is then defined as follows.

\begin{Def} \label{Def:chaosClassical}
Let $E$ be a topological space, $\mu$ be a Borel probability measure on $E$, and for every $N \in \mathbb{N}$ let $\mu_N$ be a symmetric Borel probability measure on $E^N$. 
For $k\in\N$, we say that $(\mu_N)_{N=1}^\infty$ is \textbf{$k-\mu$-chaotic} if for every choice $\phi_1,\phi_2,...,\phi_k$ of continuous bounded scalar-valued functions on $E$, we have
$$
\lim\limits_{N \rightarrow \infty} \int_{E^N} \phi_1 ( x_1) \phi_2 (x_2) \cdots  \phi_k (x_k) d\mu_N = \prod\limits_{j=1}^k \int_E \phi_j (x) d \mu (x).
$$ We say that $(\mu_N)_{N=1}^\infty$ is \textbf{$\mu$-chaotic} if $(\mu_N)_{N=1}^\infty$ is $k-\mu$-chaotic for all $k \geq 1$.
\end{Def}

Boltzmann's equation and Equation (\ref{reducedBoltzmanneq}) describe evolutions in models of classical mechanics. Corresponding quantum mechanical models are 
described in \cite[V. Quantum Mechanical Models]{Spohn}. In such models, density functions are replaced by 
\textbf{density operators}, (positive operators of trace equal to $1$), which via the trace duality define states on algebras of 
bounded linear operators acting on Hilbert spaces. The corresponding notion to the chaotic sequences of probability measures,
as well as the corresponding notion to the propagation of chaos appears in \cite[Theorem~5.7]{Spohn} where the time evolution is given by a specific family of Hamiltonians. 
Gottlieb \cite{Gottlieb} formulated the notion of chaotic sequences of density operators. In the current article, we study the notion of chaos which was introduced by Spohn and formalized by Gottlieb. To honor the fact that the definition of chaos was originated by the work of Kac for classical models, we refer to its quantum version as ``quantum Kac's chaos''. We prove two main results in this article. The first result is our Theorem~\ref{MAIN} which is the analogue of 
\cite[Proposition~2.2(i)]{Sznitman}. The second result of this article is our Theorem~\ref{PropResult} which is a simpler, different proof of the propagation of chaos result of Spohn \cite[Theorem~5.7]{Spohn}. This result shows that chaos propagates in the mean field limit for interacting quantum systems.

\textbf{Notation:} Throughout this paper, $\mathbb{H}$ will denote an arbitrary Hilbert space, $\cB(\mathbb{H})$ will denote the set of bounded operators on $\mathbb{H}$, and $\mathcal{D}(\mathbb{H})$ will denote the set of density operators on $\mathbb{H}$. The identity operator on $\cB(\mathbb{H})$ will be denoted by $1$. For any operator $A \in \cB(\mathbb{H})$ and $k \in \mathbb{N}$, $A^{\otimes k}$ will denote the tensor product of $A$ with itself $k$ times. In addition, for any $A \in \cB(\mathbb{H})$, $||A||_\infty$ will denote the $\cB(\mathbb{H})$ norm of $A$. If $A$ is a trace class operator on $\mathbb{H}$, then $||A||_1$ will denote the trace class norm of $A$.

For any $k, N \in \N$ with $k\leq N$, and $\rho_N \in \mathcal{D}(\mathbb{H}^{\otimes N})$, we will denote by $\rho_N^{(k)} \in \mathcal{D}(\mathbb{H}^{\otimes k})$ the partial trace of $\rho_N$ where we trace out all but the first $k$ copies of $\mathbb{H}$. In addition, for an index set $A \subset \{1,...,N\}$, we will denote by $\text{tr}_A(\rho_N)$ the partial trace of $\rho_N$ where we trace out the copies of $\mathbb{H}$ indexed by elements of $A$. Notice that $\text{tr}_{[k+1,N]}(\rho_N)=\rho_N^{(k)}$.

Given a topological space $E$, we will denote by $M(E)$ the set of Borel probability measures on $E$. The set of continuous bounded real-valued functions on $E$ will be denoted by $C_b(E)$. Finally, for $N\in\N$, $\Sigma_N$ will denote the set of all permutations of the set $\{1,2,...,N\}$.
\vskip.1in
{\bf Acknowledgment:} The authors would like to thank Eric Carlen for bringing the topic of quantum Kac's chaos to their attention. They also would like to thank the referee for the insightful comments
and for bringing to their attention the references \cite{BardosGolseMauser} and \cite{PaulPulvirentiSimonella}. In particular, the existence of the inequality of Corollary~\ref{referee} was conjectured by the referee.

\section{Equivalent Statements of Quantum Kac's Chaos}

Sznitman used probabilistic methods to show existence \cite{Sznitman84} and uniqueness \cite{Sznitman} to the homogeneous Boltzmann equation. The next result was important in his proofs.
\begin{Prop}\cite[Proposition $2.2$]{Sznitman}\label{EquivClassical}
Let $E$ be a separable metric space, $(\mu_N)_{N=1}^\infty$ a sequence of symmetric probability measures on $E^N$, and $\mu$ be a probability measure on $E$. The following are equivalent:
\begin{enumerate}
\item[1.] The sequence $(\mu_N)_{N=1}^\infty$ is $\mu$-chaotic (as in Definition~\ref{Def:chaosClassical}).
\item[2.] The function $X_N:E^N \rightarrow M(E)$ defined by $X_N(x_1,...,x_N)=\dfrac{1}{N}\sum\limits_{i=1}^N \delta_{x_i}$ (where $\delta_x$ stands for the Dirac measure at $x$), converges in law with respect to $\mu_N$ to the constant random variable $\mu$, i.e. for every $g \in C_b(E)$ we have that
$$\int_{E^N}|(X_N-\mu)g|^2 d\mu_N \xrightarrow[N \rightarrow \infty]{}0.$$
\item[3.] The sequence $(\mu_N)_{N=1}^\infty$ is $2-\mu$-chaotic (as in Definition~\ref{Def:chaosClassical}).
\end{enumerate}
\end{Prop} The main result of this section is to obtain a quantum analogue of Proposition~\ref{EquivClassical}. Instead of considering probability density functions, we consider density operators. We first have to extend the definition of symmetric measures (Definition~\ref{Def:symmetricClassical}) to density operators. The following is the quantum version of symmetry (Definition~\ref{Def:symmetricClassical}) we will use in this paper.
\begin{Def}\label{Symm}
Let $N \in \mathbb{N}$. A density operator $\rho_N \in \mathcal{D}(\mathbb{H}^{\otimes N})$ is \textbf{symmetric} if and only if for every 
$A_1,...,A_N \in \cB(\mathbb{H})$ and for every permutation $\pi\in\Sigma_N$,
\begin{eqnarray*}
\text{tr}(A_1 \otimes \cdot \cdot \cdot \otimes A_N \rho_N) = 
\text{tr}(A_{\pi(1)} \otimes \cdot \cdot \cdot \otimes A_{\pi(N)} \rho_N).
\end{eqnarray*}
\end{Def}
This is not the same formulation of the definition of symmetric density operators given by Gottlieb~\cite{Gottlieb}. To obtain the formulation given by Gottlieb~\cite{Gottlieb}, for $N \in \N$, define for each $\pi \in \Sigma_N$ the unitary operator $U_\pi^{[N]}\in\cB(H^{\otimes N})$ by 
\begin{eqnarray}\label{unitarty}
U^{[N]}_\pi(x_1\otimes\cdots\otimes x_N)=x_{\pi^{-1}(1)}\otimes\cdots\otimes x_{\pi^{-1}(N)}.
\end{eqnarray} A density operator $\rho_N \in \cB(\mathbb{H}^{\otimes N})$ is symmetric according to \cite{Gottlieb} if and only if $U_\pi^{[N]} \rho_N=\rho_N U_\pi^{[N]}$ for every $\pi \in \Sigma_N$.
However, Gottlieb's definition of symmetric densities is equivalent to Definition~\ref{Symm} as we show next.

\begin{Prop}\label{SymmEquivGottlieb}
Let $N \in \mathbb{N}$ and $\rho_N \in \mathcal{D}(\mathbb{H}^{\otimes N})$. Then $\rho_N$ is symmetric (as in Definition~\ref{Symm}) if and only if $U_\pi^{[N]} \rho_N = \rho_N U_\pi^{[N]}$ for all $\pi \in \Sigma_N$.
\end{Prop}
\begin{proof}
($\Rightarrow$) Let $\pi \in \Sigma_N$. Then 
\begin{eqnarray}
\nonumber\text{tr}(A_1\otimes \cdots \otimes A_N \rho_N) &=& \text{tr}(A_{\pi(1)}\otimes\cdots\otimes A_{\pi(N)}\rho_N)=\text{tr}(U_{\pi^{-1}}^{[N]}(A_1\otimes \cdots \otimes A_N)U_\pi^{[N]} \rho_N)\\
\label{break}&=&\text{tr}((A_1\otimes \cdots \otimes A_N)U_\pi^{[N]} \rho_N U_{\pi^{-1}}^{[N]})
\end{eqnarray} for any $A_1,...,A_N \in \cB(\mathbb{H}^{\otimes N})$. 

Let $(e_i)_{i\in I}$ be an orthonormal basis of $\mathbb{H}$. Then $(e_{i_1}\otimes e_{i_2}\otimes\cdots\otimes e_{i_N})_{i_1,i_2,...,i_N \in I}$ is an orthonormal basis of $\mathbb{H}^{\otimes N}$. Now assume that $U_{\pi}^{[N]}\rho_N U_{\pi^{-1}}^{[N]} \neq \rho_N$. This implies that for some $(j_1,...,j_k)\in I\times\cdots\times I$ we have $U_{\pi}^{[N]}\rho_N U_{\pi^{-1}}^{[N]}\left(e_{j_1}\otimes\cdots\otimes e_{j_N}\right) \neq \rho_N\left(e_{j_1}\otimes\cdots\otimes e_{j_N}\right)$, hence there exists $(k_1,...,k_N)\in I\times  \cdots \times I$ such that
\begin{eqnarray*}
\langle e_{k_1}\otimes\cdots\otimes e_{k_N},U_{\pi}^{[N]}\rho_N U_{\pi^{-1}}^{[N]}e_{j_1}\otimes\cdots\otimes e_{j_N}\rangle\neq\langle e_{k_1}\otimes\cdots\otimes e_{k_N},\rho_N e_{j_1}\otimes\cdots\otimes e_{j_N}\rangle.
\end{eqnarray*} This implies that
\begin{eqnarray*}
&&\text{tr}\left(U_{\pi}^{[N]}\rho_N U_{\pi^{-1}}^{[N]}\left.\left|e_{j_1}\right>\right.\left.\left<e_{k_1}\right|\right.\otimes\left.\left|e_{j_2}\right>\right.\left.\left<e_{k_2}\right|\right.\otimes\cdots\otimes\left.\left|e_{j_N}\right>\right.\left.\left<e_{k_N}\right|\right.\right)\\
&&=\langle e_{k_1}\otimes\cdots\otimes e_{k_N},U_{\pi}^{[N]}\rho_N U_{\pi^{-1}}^{[N]}e_{j_1}\otimes\cdots\otimes e_{j_N}\rangle\neq\langle e_{k_1}\otimes\cdots\otimes e_{k_N},\rho_N e_{j_1}\otimes\cdots\otimes e_{j_N}\rangle\\
&&=\text{tr}\left(\rho_N \left.\left|e_{j_1}\right>\right.\left.\left<e_{k_1}\right|\right.\otimes\left.\left|e_{j_2}\right>\right.\left.\left<e_{k_2}\right|\right.\otimes\cdots\otimes\left.\left|e_{j_N}\right>\right.\left.\left<e_{k_N}\right|\right.\right)
\end{eqnarray*} which contradicts Equation~(\ref{break}). Therefore $U_{\pi^{-1}}^{[N]}\rho_N U_{\pi}^{[N]}=\rho_N$.

($\Leftarrow$) For each $\pi \in \Sigma_N$,
\begin{eqnarray*}
\text{tr}(A_1 \otimes\cdots\otimes A_N \rho_N) &=& \text{tr}(A_1 \otimes\cdots\otimes A_N U_\pi^{[N]} \rho_N U_{\pi^{-1}}^{[N]})=\text{tr}(U_{\pi^{-1}}^{[N]}(A_1 \otimes\cdots\otimes A_N)U_\pi^{[N]} \rho_N)\\
&=& \text{tr}(A_{\pi(1)}\otimes\cdots\otimes A_{\pi(N)}\rho_N).
\end{eqnarray*}
\end{proof}
Some examples of symmetric density operators are as follows.

\begin{Ex} 
Let $\rho \in \mathcal{D}(\mathbb{H})$. 
For any $N \in \N$, define $\rho_N := \rho^{\otimes N}$. It is clear that $\rho_N$ is symmetric.
\end{Ex}

\begin{Ex}
Let $N \in \N$ and $B_1,...,B_N \in \mathcal{D}(\mathbb{H})$. Then 
\begin{eqnarray*}
\rho_N:=\dfrac{1}{N!}\sum\limits_{\sigma \in \Sigma_N} B_{\sigma(1)} \otimes \cdot \cdot \cdot \otimes B_{\sigma(N)} \in \mathcal{D}(\mathbb{H}^{\otimes N})
\end{eqnarray*} is symmetric.
\end{Ex}

\begin{Ex}\label{Ex:int}
Let $E$ be a separable topological space, $D:E\rightarrow\mathcal{D}(\mathbb{H})$ be continuous, $N\in\mathbb{N}$, and $\mu_N\in M(E^N)$ be symmetric. Then
\begin{eqnarray*}
D_N:=\int_{E^N}D(\omega_1)\otimes D(\omega_2) \otimes\cdots\otimes D(\omega_N)d\mu_N(\omega_1,\omega_2,...,\omega_N)
\end{eqnarray*} exists as a Bochner integral and $D_N \in \mathcal{D}(\mathbb{H}^{\otimes N})$ is symmetric.
\end{Ex}
\begin{proof}
For each $\pi\in\Sigma_N$ and $A_1,...,A_N \in \cB(\mathbb{H})$,
\begin{eqnarray*}
\text{tr}(A_1\otimes\cdots\otimes A_N D_N)&=&\text{tr}(A_1\otimes\cdots\otimes A_N \int_{E^N}D(\omega_1)\otimes D(\omega_2) \otimes\cdots\otimes D(\omega_N)d\mu_N)\\
&=&\int_{E^N}\text{tr}(A_1 D(\omega_1))\text{tr}(A_2 D(\omega_2))\cdots\text{tr}(A_N D(\omega_N))d\mu_N\\
&=&\int_{E^N}\text{tr}(A_1 D(\omega_{\pi^{-1}(1)})\text{tr}(A_2 D(\omega_{\pi^{-1}(2)}))\cdots\text{tr}(A_N D(\omega_{\pi^{-1}(N)}))d\mu_N\\
&=&\int_{E^N}\text{tr}(A_{\pi(1)} D(\omega_1))\text{tr}(A_{\pi(2)} D(\omega_2))\cdots\text{tr}(A_{\pi(N)} D(\omega_N))d\mu_N\\
&=&\text{tr}(A_{\pi(1)}\otimes\cdots\otimes A_{\pi(N)} \int_{E^N}D(\omega_1)\otimes D(\omega_2) \otimes\cdots\otimes D(\omega_N)d\mu_N)\\
&=&\text{tr}(A_{\pi(1)}\otimes\cdots\otimes A_{\pi(N)} D_N).
\end{eqnarray*}
\end{proof}

The following is the quantum version of Definition~\ref{Def:chaosClassical} that we will use in this paper.

\begin{Def}\label{Def:chaosQuantum}
Let $(\rho_N)_{N=1}^\infty$ be a sequence of symmetric density operators such that $\rho_N \in \mathcal{D}(\mathbb{H}^{\otimes N})$ for each 
$N \in \N$, $\rho \in \mathcal{D}(\mathbb{H})$ be a density operator, and $k \in \N$. 
Then $(\rho_N)_{N=1}^\infty$ is \textbf{$k-\rho$-chaotic} if and only if for all $A_1,...,A_k \in \cB(\mathbb{H})$,
\begin{eqnarray}\label{ChaosQuantumConvergence}
\text{tr}(A_1 \otimes \cdot \cdot \cdot \otimes A_k \otimes 1^{\otimes (N-k)} \rho_N) \xrightarrow[N \rightarrow \infty]{} 
\prod\limits_{j=1}^k \text{tr}(\rho A_j).
\end{eqnarray}
We say that $(\rho_N)_{N=1}^\infty$ is \textbf{$\rho$-chaotic} if and only if $(\rho_N)_{N=1}^\infty$ is $k-\rho$-chaotic for all $k \geq 1$.
\end{Def}

Next we give many equivalent formulations of this definition. We will use the fact that the partial trace of a density operator is a density operator.

\begin{Prop}\label{EquivDefChaos}
Let $(\rho_N)_{N=1}^\infty$ be a sequence of symmetric density matrices such that $\rho_N \in \mathcal{D}(\mathbb{H}^{\otimes N})$ for each 
$N \in \N$, $\rho \in \mathcal{D}(\mathbb{H})$, and $k \in \N$. The following are equivalent
\begin{enumerate}
\item $(\rho_N)_{N=1}^\infty$ is $k-\rho$-chaotic,
\item $\text{tr}\left(\left(\rho_N^{(k)}-\rho^{\otimes k}\right)A_1 \otimes \cdots \otimes A_k\right)\xrightarrow[N\rightarrow\infty]{}0$ for all $A_1,...,A_k \in \cB(\mathbb{H})$,
\item $\langle s, \left(\rho_N^{(k)}-\rho^{\otimes k}\right)t\rangle \xrightarrow[N \rightarrow \infty]{}0$ for all $s,t\in\mathbb{H}^{\otimes k}$, and
\label{WOTequiv}\item $\text{tr}|\rho_N^{(k)}-\rho^{\otimes k}|\xrightarrow[N \rightarrow \infty]{}0$.\label{chaosNorm}
\end{enumerate}
\end{Prop}
\begin{proof}
($(1) \Leftrightarrow (2)$) This is obvious since $\text{tr}(A_1\otimes  \cdots \otimes A_k \rho^{\otimes k})=\prod\limits_{j=1}^k\text{tr}(\rho A_j)$, and by Attal~\cite[Theorem 2.28]{Attal}, $\text{tr}(A_1\otimes\cdots\otimes A_k \otimes 1^{\otimes (N-k)}\rho_N)=\text{tr}(A_1\otimes\cdots\otimes A_k \rho_N^{(k)})$.\\
($(2) \Rightarrow (3)$) Let $\epsilon>0$, and $t,s\in\mathbb{H}^{\otimes k}$. Without loss of generality, we can make $\left|\left|t\right|\right|=\left|\left|s\right|\right|=1$. Fix an orthonormal basis $(e_i)_{i\in I}$ of $\mathbb{H}$. Then $(e_{i_1}\otimes e_{i_2}\otimes\cdots\otimes e_{i_k})_{i_1,i_2,...,i_k\in I}$ is an orthonormal basis of $\mathbb{H}^{\otimes k}$, and we can write
\begin{eqnarray*}
t=\sum\limits_{i_1,i_2,...,i_k\in I} t_{i_1 i_2 ... i_k} e_{i_1}\otimes e_{i_2} \otimes\cdots\otimes e_{i_k}\,\,\,\text{ and }\,\,\, s=\sum\limits_{j_1,j_2,...,j_k\in I} s_{j_1 j_2 ... j_k}e_{j_1}\otimes e_{j_2}\otimes\cdots\otimes e_{j_k}
\end{eqnarray*} with
\begin{eqnarray*}
\sum\limits_{i_1,i_2,...,i_k\in I} \left|t_{i_1 i_2 ... i_k}\right|^2 = 1\,\,\,\text{ and }\,\,\,\sum\limits_{j_1,j_2,...,j_k\in I} \left|s_{j_1 j_2 ... j_k}\right|^2 = 1.
\end{eqnarray*}
For every finite index set $J\subset I$, let
\begin{eqnarray*}
t_{I\setminus J}&=&\sum\limits_{i_1,i_2...,i_k\in I\setminus J} t_{i_1 i_2 ... i_k} e_{i_1}\otimes e_{i_2} \otimes\cdots\otimes e_{i_k},\\
t_{J}&=&\sum\limits_{i_1,i_2,...,i_k\in J} t_{i_1 i_2 ... i_k} e_{i_1}\otimes e_{i_2} \otimes\cdots\otimes e_{i_k}
\end{eqnarray*} and similarly define $s_{I\setminus J}$ and $s_{J}$.

Choose a finite index set $J\subset I$ such that $\left|\left|t_{I\setminus J}\right|\right|<\epsilon$ and $\left|\left|s_{I\setminus J}\right|\right|<\epsilon$. Then we have
\begin{eqnarray}
\nonumber&&\langle t,\left(\rho_N^{(k)}-\rho^{\otimes k}\right)s\rangle\\
\label{Easy1}&&=\langle t_{J},\left(\rho_N^{(k)}-\rho^{\otimes k}\right)s_{J}\rangle\\
\label{Easy2}&&+ \langle t_{J},\left(\rho_N^{(k)}-\rho^{\otimes k}\right)s_{I\setminus J}\rangle+\langle t_{I\setminus J},\left(\rho_N^{(k)}-\rho^{\otimes k}\right)s_{J}\rangle+\langle t_{I\setminus J},\left(\rho_N^{(k)}-\rho^{\otimes k}\right)s_{I\setminus J}\rangle.
\end{eqnarray} We have that $\left|\left|\rho_N^{k}-\rho^{\otimes k}\right|\right|_{\infty}\leq\left|\left|\rho_N^{k}-\rho^{\otimes k}\right|\right|_1 \leq 2$. Thus, using Cauchy-Schwarz inequalities, line~(\ref{Easy2}) is bounded independent of $N$ by
\begin{eqnarray*}
&&\left|\left|t_{J}\right|\right|\left|\left|\left(\rho_N^{(k)}-\rho^{\otimes k}\right)\right|\right|_{\infty}\left|\left|s_{I\setminus J}\right|\right|+\left|\left| t_{I\setminus J}\right|\right|\left|\left|\left(\rho_N^{(k)}-\rho^{\otimes k}\right)\right|\right|_{\infty}\left|\left|s_{J}\right|\right|\\
&&+\left|\left| t_{I\setminus J}\right|\right|\left|\left|\left(\rho_N^{(k)}-\rho^{\otimes k}\right)\right|\right|_{\infty} \left|\left| s_{I \setminus J}\right|\right| < 2\epsilon + 2\epsilon + 2\epsilon^2 = 4\epsilon + 2\epsilon^2.
\end{eqnarray*}
Line~(\ref{Easy1}) is equal to
\begin{eqnarray}
\nonumber&&\sum\limits_{\substack{i_1,i_2,...,i_k,\\j_1,j_2,...,j_k\in J}} \overline{t_{i_1 i_2 ... i_k}}s_{j_1 j_2 ... j_k} \langle e_{i_1}\otimes e_{i_2}\otimes\cdots\otimes e_{i_k},\left(\rho_N^{(k)}-\rho^{\otimes k}\right)e_{j_1}\otimes e_{j_2}\otimes\cdots\otimes e_{j_k}\rangle\\
\nonumber&&=\sum\limits_{\substack{i_1,i_2,...,i_k,\\j_1,j_2,...,j_k\in J}} \overline{t_{i_1 i_2 ... i_k}}s_{j_1 j_2 ... j_k}\text{tr}\left(\left(\rho_N^{(k)}-\rho^{\otimes k}\right)\left.\left|e_{j_1}\otimes e_{j_2}\otimes\cdots\otimes e_{j_k}\right>\right.\left.\left<e_{i_1}\otimes e_{i_2}\otimes\cdots\otimes e_{i_k}\right|\right.\right)\\
\label{Easy3}&&=\sum\limits_{\substack{i_1,i_2,...,i_k,\\j_1,j_2,...,j_k\in J}} \overline{t_{i_1 i_2 ... i_k}}s_{j_1 j_2 ... j_k}\text{tr}\left(\left(\rho_N^{(k)}-\rho^{\otimes k}\right)\left.\left|e_{j_1}\right>\right.\left.\left<e_{i_1}\right|\right.\otimes\left.\left|e_{j_2}\right>\right.\left.\left<e_{i_2}\right|\right.\otimes\cdots\otimes\left.\left|e_{j_k}\right>\right.\left.\left<e_{i_k}\right|\right.\right).
\end{eqnarray} By assumption, there exists an $N_1\in\mathbb{N}$ such that for all $N\geq N_1$, line~(\ref{Easy3}) is bounded by $\epsilon$. Therefore, for every $N\geq N_1$,
\begin{eqnarray*}
\left|\langle t,\left(\rho_N^{(k)}-\rho^{\otimes k}\right)s\rangle\right| < \epsilon + 4\epsilon + 2\epsilon^2 = 5\epsilon + 2\epsilon^2.
\end{eqnarray*}
($(3) \Leftrightarrow (4)$) Wehrl~\cite[Theorem 3]{Wehrl} proved that a sequence, $(D_N)_{N=1}^\infty \subset \mathcal{D}(\mathbb{K})$, of density operators on a Hilbert space $\mathbb{K}$ converges in the weak operator topology to a density operator $D \in \mathcal{D}(\mathbb{K})$ if and only if it converges in norm, i.e. $\text{tr}|D_N-D|\xrightarrow[N\rightarrow\infty]{}0$. Our result follows by letting $\mathbb{K}=\mathbb{H}^{\otimes k}$, $D_N=\rho_N^{(k)}$ for each $N$, and $D=\rho^{\otimes k}$. (Wehrl~\cite{Wehrl} assumed in his paper that all Hilbert spaces considered are separable. However, the examination of his proof shows that the assumption of separability of the Hilbert space is not needed in Theorem 3.)\\
($(4)\Rightarrow (2)$) This is obvious.
\end{proof}
\begin{Rmk} Notice that condition~(\ref{WOTequiv}) of Proposition~\ref{EquivDefChaos} is equivalent to
\begin{eqnarray}\label{finiterank}
\text{tr}\left(\left(\rho_N^{(k)}-\rho^{\otimes k}\right)A\right)\xrightarrow[N\rightarrow\infty]{}0 \text{ for all finite rank operators } A\in\cB(\mathbb{H}^{\otimes k}).
\end{eqnarray} Furthermore since (\ref{finiterank}) is weaker than
\begin{eqnarray}\label{nofiniterank}
\text{tr}\left(\left(\rho_N^{(k)}-\rho^{\otimes k}\right)A\right)\xrightarrow[N\rightarrow\infty]{}0 \text{ for all } A\in\cB(\mathbb{H}^{\otimes k})
\end{eqnarray} which is further weaker than condition~(\ref{chaosNorm}) of Proposition~\ref{EquivDefChaos}, we have that conditions (\ref{finiterank}) and (\ref{nofiniterank}) are equivalent to the conditions presented in Proposition~\ref{EquivDefChaos}.
\end{Rmk}

Condition (\ref{chaosNorm}) in Proposition~\ref{EquivDefChaos} appears in \cite[Theorem 5.7]{Spohn}. Gottlieb \cite{Gottlieb} used this condition for all $k$ to define that ``$\rho_N$ is $\rho$-chaotic''. Proposition~\ref{EquivDefChaos} shows that Gottlieb's definition agrees with ours. We will now give some examples of chaotic sequences.

\begin{Ex}
Let $\rho \in \cB(\mathbb{H})$. For each $N \in \N$, define $\rho_N:=\rho^{\otimes N}$. Then it is clear that 
$(\rho_N)_{N=1}^\infty$ is $\rho$-chaotic.
\end{Ex}

The following example due to Gottlieb~\cite[Lemma 1.3.2]{Gottlieb} gives a way of constructing a chaotic sequence of density operators from any classically chaotic sequence of probability measures.

\begin{Ex}
Let $E$ be a separable topological space, and for each $N\in\mathbb{N}$ let $\mu_N\in M(E^N)$ and $\mu\in M(E)$ such that $(\mu_N)_{N=1}^\infty$ is $\mu$-chaotic. Let $D:E\rightarrow\mathcal{D}(\mathbb{H})$ be a continuous function. Then for $N\in\mathbb{N}$, $D_N:=\int_{E^N} D(\omega_1) \otimes D(\omega_2) \otimes \cdot \cdot \cdot \otimes D(\omega_N) d\mu_N(\omega_1,\omega_2,...,\omega_N)$ 
and $\overline{D}:=\int_E D(\omega) d\mu$ exist as Bochner integrals, $D_N \in\mathcal{D}(\mathbb{H}^{\otimes N})$ is symmetric, $\overline{D}\in\mathcal{D}(\mathbb{H})$, and $(D_N)_{N=1}^\infty$ is $\overline{D}$-chaotic.
\end{Ex}
\begin{proof}
For any 
$k \geq 1$ and $A_1,...,A_k \in \cB(\mathbb{H})$, 
\begin{eqnarray*}
&&\text{tr}(A_1 \otimes \cdot \cdot \cdot \otimes A_k \otimes 1^{\otimes (N-k)} 
\int D(\omega_1) \otimes \cdot \cdot \cdot \otimes D(\omega_N) d\mu_N)\\
&=& \text{tr}\left(\int A_1 D(\omega_1) \otimes \cdot \cdot \cdot \otimes A_k D(\omega_k) \otimes D(\omega_{k+1}) \otimes \cdot \cdot \cdot \otimes D(\omega_N) d\mu_N\right)\\
&=& \int \text{tr}(A_1 D(\omega_1)) \cdot \cdot \cdot \text{tr}(A_k D(\omega_k)) d\mu_N
\end{eqnarray*} which converges to
$$\int \text{tr}(A_1 D(\omega_1)) \cdot \cdot \cdot \text{tr}(A_k D(\omega_k)) d\mu^{\otimes k} = \prod_{j=1}^k \text{tr}\left(A_j \int D(\omega) d\mu\right)$$ as $N$ approaches infinity by Definition~\ref{Def:chaosClassical}.
\end{proof}

Now we are ready to prove the analogous statement to Proposition~\ref{EquivClassical} (\cite[Proposition $2.2$]{Sznitman}) for chaotic sequences of density operators.

\begin{Thm}\label{MAIN}
Let $(\rho_N)_{N=1}^\infty$ be a sequence of symmetric density operators such that $\rho_N \in \mathcal{D}(\mathbb{H}^{\otimes N})$ for each $N \in \N$, 
and let $\rho \in \mathcal{D}(\mathbb{H})$. Then the following are equivalent.
\begin{enumerate}
\item[(1)] $(\rho_N)_{N=1}^\infty$ is $k-\rho$-chaotic for all $k \in \mathbb{N}$,
\item[(2)] $(\rho_N)_{N=1}^\infty$ is $2-\rho$-chaotic, and
\item[(3)] for each $A \in \cB(\mathbb{H})$,
\begin{eqnarray}\label{EmpiricalMeasureConvergence}
\text{tr}\left(\left|\dfrac{1}{N} \sum\limits_{j=1}^N 1^{\otimes (j-1)} \otimes A \otimes 1^{\otimes (N-j)} - \text{tr}(A\rho) 1^{\otimes N}\right|^2\rho_N\right) \xrightarrow[N \rightarrow \infty]{} 0.
\end{eqnarray}
\end{enumerate}
\end{Thm}

\begin{Def}
For $N\in\mathbb{N}$, define $X_N:\cB(\mathbb{H})\rightarrow\cB(\mathbb{H}^{\otimes N})$ by
\begin{eqnarray*}
X_N(A)=\dfrac{1}{N}\sum\limits_{j=1}^N 1^{\otimes (j-1)}\otimes A \otimes 1^{\otimes (N-j)}.
\end{eqnarray*}
\end{Def}
The function $X_N$ is studied in \cite{GolsePaul} and is called a \textbf{quantum empirical measure}. The above theorem and \cite[Proposition $2.2$]{Sznitman} gives more justifications for the choice of this term.

\begin{proof}
($(1) \Rightarrow (2)$) This is obvious.

($(2) \Rightarrow (3)$) Let $A \in \cB(\mathbb{H})$. Notice that
\begin{eqnarray}
\label{2imp3}&\text{tr}\left(\left|\dfrac{1}{N} \sum\limits_{j=1}^N 1^{\otimes (j-1)} \otimes A \otimes 1^{\otimes (N-j)} - \text{tr}(A\rho) 1^{\otimes N}\right|^2\rho_N\right)
\end{eqnarray}
$$= \text{tr}((\dfrac{1}{N} \sum\limits_{i=1}^N 1^{\otimes (i-1)} \otimes A^* \otimes 1^{\otimes (N-i)} - \overline{\text{tr}(A\rho)}1^{\otimes N})(\dfrac{1}{N} \sum\limits_{j=1}^N 1^{\otimes (j-1)} \otimes A \otimes 1^{\otimes (N-j)} - \text{tr}(A\rho)1^{\otimes N})\rho_N).$$
By distributing, we obtain that the last expression is equal to
\begin{align}
&\label{first}\dfrac{1}{N^2} \text{tr}(\sum\limits_{i,j=1}^N (1^{\otimes (i-1)}\otimes A^* \otimes 1^{\otimes (N-i)})(1^{\otimes (j-1)}\otimes A \otimes 1^{\otimes (N-j)}) \rho_N)\\
\label{second}&- \dfrac{\text{tr}(A\rho)}{N} \text{tr}(\sum\limits_{j=1}^N 1^{\otimes (j-1)} \otimes A^* \otimes 1^{\otimes (N-j)} \rho_N)\\\label{third}
&- \dfrac{\overline{\text{tr}(A\rho)}}{N} \text{tr}(\sum\limits_{j=1}^N 1^{\otimes (j-1)} \otimes A \otimes 1^{\otimes (N-j)} \rho_N)\\\label{fourth}
&+|\text{tr}(A\rho)|^2.
\end{align}
We will obtain that the sum of lines (\ref{first}), (\ref{second}), (\ref{third}), and (\ref{fourth}) goes to zero as $N$ approaches infinity. To evaluate (\ref{first}), we consider three cases: when $i=j$, when $i < j$, and when $j < i$. If $i=j$, then by symmetry of $\rho_N$,
\begin{eqnarray*}
\dfrac{1}{N^2} \sum\limits_{j=1}^N \text{tr}( 1^{\otimes (j-1)} \otimes |A|^2 \otimes 1^{\otimes (N-j)} \rho_N) &=& \dfrac{1}{N} \text{tr}(|A|^2 \otimes 1^{\otimes (N-1)} \rho_N)\\
&\leq& \dfrac{1}{N} |||A|^2 \otimes 1^{\otimes (N-1)}||_\infty ||\rho_N||_1 \leq \dfrac{|||A|^2||_\infty}{N},
\end{eqnarray*} which goes to zero as $N$ approaches infinity. If $i < j$, then by symmetry of $\rho_N$,
\begin{eqnarray*}
&&\dfrac{1}{N^2} \sum\limits_{i < j} \text{tr}( 1^{\otimes (i-1)} \otimes A^* \otimes 1^{\otimes (j-i-1)} \otimes A \otimes 1^{\otimes(N-j)} \rho_N)\\
&=& \dfrac{1}{N^2} \dfrac{N!}{2(N-2)!} \text{tr}(A^* \otimes A \otimes 1^{\otimes (N-2)} \rho_N)\\ &=& \dfrac{N-1}{2N} \text{tr}(A^* \otimes A \otimes 1^{\otimes (N-2)} \rho_N)
\xrightarrow[N \rightarrow \infty]{2-\rho-\text{chaotic}} \dfrac{1}{2}\text{tr}(A \rho)\text{tr}(A^* \rho) = \dfrac{1}{2}|\text{tr}(A\rho)|^2.
\end{eqnarray*}
If $j < i$ we obtain exactly the same limit. Thus, we have that line (\ref{first}) converges to $|\text{tr}(A\rho)|^2$ as $N$ approaches infinity.

Using symmetry of $\rho_N$ and by assumption, we obtain the limit of line (\ref{second}),
\begin{eqnarray*}
\dfrac{-\text{tr}(A\rho)}{N} &\text{tr}&(\sum\limits_{j=1}^N 1^{\otimes (j-1)} \otimes A^* \otimes 1^{\otimes (N-j)} \rho_N)\\
 &=& -\text{tr}(A\rho) \text{tr}(A^* \otimes 1^{\otimes (N-1)}\rho_N) \xrightarrow[N\rightarrow\infty]{} - \text{tr}(A\rho)\text{tr}(A^* \rho) = - |\text{tr}(A\rho)|^2,
\end{eqnarray*} where in the last limit, we used the obvious fact that if $\rho_N$ is $2-\rho$-chaotic then it is $1-\rho$-chaotic.
Similarly, line (\ref{third}) converges to $-|\text{tr}(A\rho)|^2$ as $N$ approaches infinity.

Therefore, the sum of lines (\ref{first}), (\ref{second}), (\ref{third}), and (\ref{fourth}) converge to 
\begin{eqnarray*}
|\text{tr}(A\rho)|^2 - |\text{tr}(A\rho)|^2 - |\text{tr}(A\rho)|^2 + |\text{tr}(A\rho)|^2 = 0,
\end{eqnarray*} and line (\ref{2imp3}) converges to $0$ as $N$ approaches infinity.

($(3) \Rightarrow (1)$) Let $k \in \mathbb{N}$ and $A_1,...,A_k \in \cB(\mathbb{H})$. Then
\begin{eqnarray}
\label{3imp1}&&\left|\text{tr}\left(A_1 \otimes \cdots \otimes A_k \otimes 1^{\otimes (N-k)}\rho_N\right) - \prod\limits_{j=1}^k \text{tr}(\rho A_j)\right|\leq\\\label{fifth}
&&\left|\text{tr}\left(A_1 \otimes \cdots \otimes A_k \otimes 1^{\otimes (N-k)} \rho_N\right)\right.\\\nonumber
&-& \left.\text{tr}\left(\prod\limits_{j=1}^k \dfrac{1}{N} \left(A_j \otimes 1^{\otimes (N-1)} + 1 \otimes A_j \otimes 1^{\otimes (N-2)} + \cdots + 1^{\otimes (N-1)} \otimes A_j\right)\rho_N\right)\right|\\\label{sixth}
&+& \left| \text{tr}\left(\prod\limits_{j=1}^k \dfrac{1}{N} \left(A_j \otimes 1^{\otimes (N-1)} + 1 \otimes A_j \otimes 1^{\otimes (N-2)} + \cdots + 1^{\otimes (N-1)} \otimes A_j\right)\rho_N\right)\right.\\\nonumber
&-& \left.\prod\limits_{j=1}^N \text{tr}(\rho A_j)\right|.
\end{eqnarray} We label the first and second lines after the inequality by (\ref{fifth}) and the third and fourth lines after the inequality by (\ref{sixth}). Our goal will be to show that the sum of lines (\ref{fifth}) and (\ref{sixth}) goes to $0$ as $N$ approaches infinity.

For lines (\ref{fifth}), for $k \leq N$ we define $E_{k,N}$ to be the set of embedings (i.e. one-to-one maps)\\ $\sigma:\{1,...,k\} \rightarrow \{1,...,N\}$. Notice that $\# E_{k,N} = \dfrac{N!}{(N-k)!}$. Furthermore, for $\sigma \in E_{k,N}$ and $i \in \{1,...,N\}$, define
\[A_{\sigma,i}:= \begin{cases} 
      A_j & \text{when } \sigma(j)=i \\
      1 & \text{otherwise}. 
   \end{cases}
\]
Then, by symmetry of $\rho_N$, we can rewrite lines (\ref{fifth}) as
\begin{eqnarray}
\label{seventh}&&\left|\text{tr}\left(\left(\dfrac{(N-k)!}{N!} \sum\limits_{\sigma \in E_{k,N}}\left[A_{\sigma,1} \otimes A_{\sigma,2} \otimes 
\cdot \cdot \cdot \otimes A_{\sigma,N}\right]\right.\right.\right.\\
\label{eighth}&&\left.\left.\left.-\dfrac{1}{N^k} 
\prod\limits_{j=1}^k \left(A_j \otimes 1^{\otimes (N-1)} + 1 \otimes A_j \otimes 1^{\otimes (N-2)} + 
\cdot \cdot \cdot + 1^{\otimes (N-1)}\otimes A_j\right)\right)\rho_N\right)\right|.
\end{eqnarray}
In line (\ref{eighth}), there are two types of terms: the terms with $N-k$ $1$'s in the expanded form which we call the off-diagonal terms, and all the other terms which we call the diagonal terms. There are $\dfrac{N!}{(N-k)!}$ off-diagonal terms and $N^k-\dfrac{N!}{(N-k)!}$ diagonal terms. Let $M:=\max\limits_{j=1,...,k} ||A_j||_\infty$. The off-diagonal terms are exactly the terms of line (\ref{seventh}). Thus, the addition of line (\ref{seventh}) and the off-diagonal terms of line (\ref{eighth}) is bounded by
\begin{eqnarray*}
\left|\left|\left(\dfrac{(N-k)!}{N!} - \dfrac{1}{N^k} \right)\sum\limits_{\sigma \in E_{k,N}} A_{\sigma,1} \otimes \cdot \cdot \cdot \otimes A_{\sigma,N}\right|\right|_\infty \left|\left|\rho_N\right|\right|_1 \leq \dfrac{N!}{(N-k)!}\left(\dfrac{(N-k)!}{N!} - \dfrac{1}{N^k}\right)M^k.
\end{eqnarray*} Each diagonal term is also bounded by $M^k$. Thus, the diagonal terms of line (\ref{eighth}) are bounded by $\dfrac{1}{N^k} \left(N^k - \dfrac{N!}{(N-k)!}\right)M^k$. Hence, we can bound lines (\ref{fifth}) and take the limit as $N$ approaches $\infty$,
\begin{eqnarray*}
&&\dfrac{N!}{(N-k)!}\left(\dfrac{(N-k)!}{N!}-\dfrac{1}{N^k}\right)M^k + \dfrac{1}{N^k}\left(N^k - 
\dfrac{N!}{(N-k)!}\right)M^k\\ 
&=& M^k\left[\left(\dfrac{N^k(N-k)!}{N!}-1\right)\dfrac{N!}{N^k(N-k)!} + \dfrac{N!}{N^k(N-k)!}\left(\dfrac{N^k(N-k)!}{N!} - 
1\right)\right]\\ 
&=& 2M^k\left[\dfrac{N!}{N^k(N-k)!}\left(\dfrac{N^k(N-k)!}{N!}-1\right)\right]=2M^k\left[1-\dfrac{N!}{N^k(N-k)!}\right] 
\xrightarrow[N \rightarrow \infty]{} 0.
\end{eqnarray*} So line (\ref{fifth}) goes to $0$ as $N$ approaches infinity.

For lines (\ref{sixth}) can be rewritten as
\begin{eqnarray*}
\nonumber&&\left|\text{tr}\left[\left(\prod\limits_{j=1}^k X_N(A_j) - \prod\limits_{j=1}^k \text{tr}(\rho A_j)1^{\otimes N}\right)\rho_N\right]\right|\\
&=&\left|\sum\limits_{l=0}^{k-1} \text{tr}\left[\left(\prod\limits_{j=1}^l \text{tr}(\rho A_j)\prod\limits_{j=l+1}^k X_N(A_j) - \prod\limits_{j=1}^{l+1} \text{tr}(\rho A_j)\prod\limits_{j=l+2}^k X_N(A_j)1^{\otimes N}\right)\rho_N\right]\right|\\
&=&\left|\sum\limits_{l=0}^{k-1} \text{tr}\left[\left(X_N(A_{l+1})-\text{tr}(\rho A_{l+1})1^{\otimes N}\right)\prod\limits_{j=1}^l \text{tr}(\rho A_j) \prod\limits_{j=l+2}^k X_N(A_j) \rho_N\right]\right|,
\end{eqnarray*} and can be bounded by
\begin{eqnarray}
\label{CorrelationErrors}&&\sum\limits_{l=0}^{k-1} \left|\text{tr}\left[\left(X_N(A_{l+1})-\text{tr}(\rho A_{l+1})1^{\otimes N}\right)\prod\limits_{j=1}^l \text{tr}(\rho A_j) \prod\limits_{j=l+2}^k X_N(A_j)\rho_N\right]\right|\\
\nonumber&\leq& \sum\limits_{l=0}^{k-1} \left(\sqrt{\text{tr}\left[\left|X_N(A_{l+1}^*)-\text{tr}(\rho A_{l+1}^*)1^{\otimes N}\right|^2\rho_N\right]}\cdot \sqrt{\text{tr}\left[\left|\prod\limits_{j=1}^l \text{tr}(\rho A_j) \prod\limits_{j=l+2}^k X_N(A_j)\right|^2 \rho_N\right]} \right)
\end{eqnarray} since by \cite[Lemma 2.3.10]{BratteliRobinson} applied to the positive linear functional $\omega(\cdot):=\text{tr}(\cdot\rho_N)$ we obtain $\left|\text{tr}\left(BC\rho_N\right)\right|\leq\text{tr}\left(\left|B^*\right|^2\rho_N\right)^{1/2}\text{tr}\left(\left|C\right|^2\rho_N\right)^{1/2}$  for any $B,C\in\cB(\mathbb{H}^{\otimes N})$. The quantities that appear inside the brackets of (\ref{CorrelationErrors}) resemble the so called ``correlation errors'' which are studied by Paul, Pulvirenti, and Simonella \cite{PaulPulvirentiSimonella} (see also references within). The rate of convergence of these quantities is called the ``size of chaos'' in \cite{PaulPulvirentiSimonella}.

By assumption, for each $l$,
\begin{eqnarray*}
\sqrt{\text{tr}\left[\left|X_N(A_{l+1}^*)-\text{tr}(\rho A_{l+1}^*)1^{\otimes N}\right|^2\rho_N\right]} \xrightarrow[N \rightarrow \infty]{} 0,
\end{eqnarray*} and if $M:=\max\limits_{j=1,...,k} ||A_j||_\infty$, since $||X_N(A_j)||_\infty \leq ||A_j||_\infty \leq M$, we have that
\begin{eqnarray*}
&&\text{tr}\left[\left|\prod\limits_{j=1}^l \text{tr}(\rho A_j) \prod\limits_{j=l+2}^k X_N(A_j)\right|^2 \rho_N\right] \leq \prod\limits_{j=1}^l |\text{tr}(\rho A_j)|^2 \text{tr}\left[\left| \prod\limits_{j=l+2}^k X_N(A_j)\right|^2 \rho_N\right]\\
&\leq& \prod\limits_{j=1}^l \left|\text{tr}(\rho A_j)\right|^2 \left|\left|\prod\limits_{j=l+2}^k X_N(A_j)\right|\right|_\infty^2 \left|\left|\rho_N\right|\right|_1 \leq M^{2k} \prod\limits_{j=1}^l \left|\text{tr}(\rho A_j)\right|^2
\end{eqnarray*} which is bounded independent of $N$. Hence, lines (\ref{sixth}) converges to $0$ as $N$ goes to infinity. Therefore, line (\ref{3imp1}) converges to $0$ as $N$ approaches infinity.
\end{proof}

We can extract from the proof of Theorem~\ref{MAIN} a bound on the rate of convergence of the limit in line~(\ref{ChaosQuantumConvergence}) given the rate the associated empirical measures convergence in line~(\ref{EmpiricalMeasureConvergence}). In other words, we are able to find a bound on how fast marginals of $\rho_N$ converge to a tensor product in the sense of 
condition $(2)$ of Proposition~\ref{EquivDefChaos} 
knowing how fast the quantum empirical measures converge. The following corollary makes this statement precise.

\begin{Cor} \label{referee}
For each $N\in\mathbb{N}$ and $A\in\cB(\mathbb{H})$, define $e_N(A)$ to be the quantity in line~(\ref{2imp3}), and for $k\in\mathbb{N}$ and each $A_1,A_2,...,A_k\in\cB(\mathbb{H})$, define $C_{k,N}(A_1,...,A_k)$ to be the quantity in line~(\ref{3imp1}). Then by examination of the proof of Theorem~\ref{MAIN}, we are able to see that
\begin{eqnarray*}
C_{k,N}(A_1,...,A_k)&\leq&\sum\limits_{l=0}^{k-1}\sqrt{e_N(A_{l+1}^*)}\left(\prod\limits_{j=1}^l\left|\text{tr}(\rho A_j)\right|^2\prod\limits_{j=l+2}^k||A_j||_{\infty}^2\right)\\
&+&2\prod\limits_{i=1}^k||A_i||_\infty\left(1-\dfrac{N!}{N^k(N-k)!}\right).
\end{eqnarray*}
\end{Cor}

The next corollary follows from Theorem~\ref{MAIN} and Proposition~\ref{EquivDefChaos}.
\begin{Cor}
Let $(\rho_N)_{N=1}^\infty$ be a sequence of symmetric density operators such that $\rho_N\in\mathcal{D}(\mathbb{H}^{\otimes N})$ for each $N \in\mathbb{N}$, and let $\rho\in\mathcal{D}(\mathbb{H})$. Then the following are equivalent.
\begin{enumerate}
\item $(\rho_N)_{N=1}^\infty$ is $\rho$-chaotic,
\item $\text{tr}\left|\rho_N^{(k)}-\rho^{\otimes k}\right|\xrightarrow[N\rightarrow\infty]{}0$ for all $k \in\mathbb{N}$, and
\item $\text{tr}\left|\rho_N^{(2)}-\rho^{\otimes 2}\right|\xrightarrow[N\rightarrow\infty]{}0$.
\end{enumerate}
\end{Cor}

\section{Propagation of Chaos}

Spohn proved that under evolutions governed by certain families of Hamiltonians, chaotic sequences of density operators propagate in time \cite[Theorem 5.7]{Spohn}. In this section, we will use the ideas of the proofs of Ducomet~\cite[Theorem 3.1]{Ducomet}, and Bardos, Golse, Gottlieb, and Mauser~\cite[Theorem 3.1]{Bardos} to give a simple, different proof to the result of Spohn. First, we define propagation of chaos.
\begin{Def}\label{PropagationChaos}
Let $(\rho_N(0))_{N=1}^\infty$ be a sequence of density operators and let $(H_N)_{N=1}^\infty$ be a sequence of Hamiltonians where $\rho_N(0) \in \mathcal{D}(\mathbb{H}^{\otimes N})$ and $H_N \in \cB(\mathbb{H}^{\otimes N})$ for every $N\in\N$. For each $t \geq 0$ and $N\in\N$, define the density operator 
\begin{eqnarray}\label{NewDensity}
\rho_N(t):=e^{-itH_N}\rho_N(0)e^{itH_N} \in \mathcal{D}(\mathbb{H}^{\otimes N}).
\end{eqnarray}
If, for each fixed $t \geq 0$, the sequence $(\rho_N(t))_{N=1}^\infty$ is $\rho(t)$-chaotic for some $\rho(t)\in\mathcal{D}(\mathbb{H})$, then we say that chaos \textbf{propagates with respect to $(H_N)_{N=1}^\infty$}.
\end{Def}

We will now construct, as in Spohn~\cite{Spohn}, examples of propagation of chaos. We will examine the mean field limit for interacting quantum particles, see \cite[pages 609 - 613]{Spohn}. For each $N\in\N$ and $\pi \in \Sigma_N$, define the unitary operator $U_\pi^{[N]} \in \cB(\mathbb{H}^{\otimes N})$ by equation (\ref{unitarty}). For $A \in \cB(\mathbb{H})$, $V \in \cB(\mathbb{H}\otimes \mathbb{H})$, $N\in\N$, and $j\in\{1,...,N\}$, define $$A_j^{[N]}:=1^{\otimes (j-1)}\otimes A \otimes 1^{\otimes (N-j-1)} \in \cB(\mathbb{H}^{\otimes N}),$$ $$V_{12}^{[N]}:=V\otimes 1^{\otimes(N-2)},$$ and $$V_{ij}^{[N]}=U_{\pi^{-1}}^{[N]}V_{12}^{[N]}U_{\pi}^{[N]}$$ where $\pi$ is any permutation where $\pi(i)=1$ and $\pi(j)=2$. Notice that this operator is well defined and independent of the permutation $\pi$ that we use, (as long as $\pi(i)=1$ and $\pi(j)=2$) because when applied to a simple tensor $x_1\otimes\cdots\otimes x_N$ all but the $x_i$ and $x_j$ spots are left invariant. For any self-adjoint $A \in \cB(\mathbb{H})$, any self-adjoint $V\in\cB(\mathbb{H}\otimes\mathbb{H})$, and each $N \in \mathbb{N}$, consider the Hamiltonian
\begin{eqnarray}\label{HamiltonianSeq}
H_N=\sum\limits_{j=1}^N A_j^{[N]} + \dfrac{1}{N}\sum\limits_{i \neq j;i,j=1}^N V_{ij}^{[N]}.
\end{eqnarray} Also, define 
\begin{eqnarray}\label{ReducedHamSeq}
H_{n,N}:=\sum\limits_{j=1}^n A_j^{[n]} + \dfrac{1}{N}\sum\limits_{i \neq j;i,j=1}^n V_{ij}^{[n]}
\end{eqnarray} for each $n,N \in \mathbb{N}$, $n \leq N$. 

The main result of this section is Theorem~\ref{PropResult}. In this theorem, we will assume that a sequence of density operators $(\rho_N(0))_{N=1}^\infty$ is $\rho(0)$-chaotic and we will show that if $(H_N)_{N=1}^\infty$ is defined by equation (\ref{HamiltonianSeq}) and for all $t\geq 0$, $(\rho_N(t))_{N=1}^\infty$ is defined by equation (\ref{NewDensity}), then for all $t\geq 0$ the sequence $(\rho_N(t))_{N=1}^\infty$ is $\rho(t)$-chaotic for some $\rho(t) \in \mathcal{D}(\mathbb{H})$, i.e. chaos propagates with respect to $(H_N)_{N=1}^\infty$. Before proving our main result (Theorem~\ref{PropResult}), we need to establish some preliminary results. The first preliminary result consists of proving that $\rho_N(t)$ is symmetric for each $N\in\N$ and $t \geq 0$.

\begin{Prop}
For each $N \in \mathbb{N}$ and $t \geq 0$, $\rho_N(t)$ (as defined in equation (\ref{NewDensity})) is symmetric.
\end{Prop}
\begin{proof}
Let $\pi \in \Sigma_N$. By Proposition~\ref{SymmEquivGottlieb}, we must show that $U_{\pi^{-1}}^{[N]}e^{-itH_N}\rho_N(0) e^{itH_N}U_\pi^{[N]} = U_{\pi^{-1}}^{[N]} \rho_N(t)U_\pi^{[N]}=\rho_N(t)$. Since $\rho_N(0)$ is symmetric, it is enough to show that $U_{\pi^{-1}}^{[N]}e^{itH_N}U_\pi^{[N]}=e^{itH_N}$. Furthermore, it is enough to show that $U_{\pi^{-1}}^{[N]}H_NU_\pi^{[N]} = H_N$.

First we prove that $U_{\pi^{-1}}^{[N]}\sum\limits_{j=1}^N A_j^{[N]} U_\pi^{[N]} = \sum\limits_{j=1}^N A_j^{[N]}$. Indeed,
\begin{eqnarray*}
&&U_{\pi^{-1}}^{[N]} \sum\limits_{j=1}^N A_j^{[N]}U_\pi^{[N]}(x_1 \otimes \cdots \otimes x_N)=\sum\limits_{j=1}^N U_{\pi^{-1}}^{[N]}A_j^{[N]}(x_{\pi^{-1}(1)}\otimes \cdots \otimes x_{\pi^{-1}(N)})\\
&&= \sum\limits_{j=1}^N U_{\pi^{-1}}^{[N]}(x_{\pi^{-1}(1)} \otimes \cdots \otimes x_{\pi^{-1}(j-1)} \otimes A(x_{\pi^{-1}(j)}) \otimes x_{\pi^{-1}(j+1)} \otimes \cdots \otimes x_{\pi^{-1}(N)})\\
&&=\sum\limits_{j=1}^N x_1 \otimes \cdots \otimes x_{\pi^{-1}(j)-1} \otimes A(x_{\pi^{-1}(j)}) \otimes x_{\pi^{-1}(j)+1} \otimes \cdots \otimes x_N\\
&&= \sum\limits_{j=1}^N x_1 \otimes \cdots \otimes x_{j-1} \otimes Ax_j \otimes x_{j+1} \otimes \cdots \otimes x_N = \sum\limits_{j=1}^N A_j^{[N]}(x_1 \otimes \cdots \otimes x_N).
\end{eqnarray*}

Next, will will show that $U_{\pi^{-1}}^{[N]}\sum\limits_{i\neq j;i,j = 1}^N V_{ij}^{[N]}U_{\pi}^{[N]}=\sum\limits_{i\neq j;i,j = 1}^N V_{ij}^{[N]}$. For each $i,j\in\{1,...,N\}$ with $i \neq j$, choose $\sigma_{ij}\in\Sigma_N$ with $\sigma_{ij}(i)=1$ and $\sigma_{ij}(j)=2$. Then
\begin{eqnarray}
\nonumber U_{\pi^{-1}}^{[N]}\sum\limits_{i\neq j;i,j=1}^N V_{ij}^{[N]} U_\pi^{[N]} &=& \sum\limits_{i\neq j;i,j=1}^N U_{\pi^{-1}}^{[N]} V_{ij}^{[N]}U_\pi^{[N]} = \sum\limits_{i \neq j;i,j=1}^N U_{\pi^{-1}}^{[N]} U_{\sigma_{ij}^{-1}}^{[N]}V_{12}^{[N]}U_{\sigma_{ij}}^{[N]} U_\pi^{[N]}\\
\label{newperm}&=& \sum\limits_{i\neq j;i,j=1}^N U_{(\sigma_{ij} \pi)^{-1}}^{[N]} V_{12}^{[N]} U_{\sigma_{ij}\pi}^{[N]}.
\end{eqnarray} Notice that $(\sigma_{ij}\pi)(\pi^{-1}(i))=1$ and $(\sigma_{ij}\pi)(\pi^{-1}(j))=2$, and thus, line (\ref{newperm}) is equal to
\begin{eqnarray*}
\sum\limits_{i\neq j;i,j=1}^N V_{\pi^{-1}(i) \pi^{-1}(j)}^{[N]} = \sum\limits_{i \neq j;i,j=1}^N V_{ij}^{[N]},
\end{eqnarray*} where the last equality is valid because $i\neq j$ if and only if $\pi^{-1}(i)\neq\pi^{-1}(j)$.
Thus, we obtain that $U_{\pi^{-1}}^{[N]}\sum\limits_{i\neq j = 1}^N V_{ij}^{[N]}U_{\pi}^{[N]}=\sum\limits_{i\neq j = 1}^N V_{ij}^{[N]}$. Hence, we have that $U_{\pi^{-1}}^{[N]}H_NU_\pi^{[N]}=H_N$, and $\rho_N(t)$ is symmetric.
\end{proof}

We are aiming to construct two similar families of differential equations for $\left(\rho_N^{(n)}(t)\right)_{n=1}^{N-1}$ and $\left(\rho(t)^{\otimes n}\right)_{n=1}^\infty$. The following Proposition gives a family of differential equations which is satisfied by $\left(\rho_N^{(n)}(t)\right)_{n=1}^{N-1}$.
\begin{Prop}\label{Hartree1}
Let $N\in\mathbb{N}$. For $n\in\mathbb{N}$, $n\leq N-1$, and $t\geq 0$, we have
\begin{eqnarray}\label{rhoNDiffEq}
i\dfrac{d}{dt}\rho_N^{(n)}(t)=[H_{n,N},\rho_N^{(n)}(t)]+\dfrac{N-n}{N}\sum\limits_{j=1}^n \text{tr}_{\{n+1\}}[V_{j \,n+1}^{[n+1]}+V_{n+1\, j}^{[n+1]},\rho_{N}^{(n+1)}(t)]
\end{eqnarray} where $\rho_N(t)$ is given by (\ref{NewDensity}) and $H_{n,N}$ is given by (\ref{ReducedHamSeq}).
\end{Prop}
\begin{proof}
We know
\begin{eqnarray*}
i\dfrac{d}{dt}\rho_N(t)=\left[H_N,\rho_N(t)\right].
\end{eqnarray*} 
Taking the partial trace of both sides, and using the fact that partial traces and derivatives commute, we obtain
\begin{eqnarray}\label{DiffEq1}
i\dfrac{d}{dt}\rho_N^{(n)}(t)=\text{tr}_{[n+1,N]}\left[H_N,\rho_N(t)\right].
\end{eqnarray} We claim that 
\begin{eqnarray}
\nonumber&&\text{tr}_{[n+1,N]}\left(\left[H_N,\rho_N(t)\right]\right)\\
\label{claim1}&&=[H_{n,N},\rho_N^{(n)}(t)]+\dfrac{N-n}{N}\sum\limits_{j=1}^n \text{tr}_{\{n+1\}}[V_{j \,n+1}^{[n+1]}+V_{n+1\, j}^{[n+1]},\rho_{N}^{(n+1)}(t)].
\end{eqnarray} In order to prove equation~(\ref{claim1}), by line \cite[equation (2.11)]{Attal}, we need to prove that for every $B \in \cB(\mathbb{H}^{\otimes n})$
\begin{eqnarray*}
\text{tr}\left([H_N,\rho_N(t)]B\otimes 1^{\otimes (N-n)}\right)&=&\text{tr}\left(\left[\sum\limits_{j=1}^n A_j^{[n]} + \dfrac{1}{N}\sum\limits_{i\neq j = 1}^n V_{ij}^{[n]},\rho_N^{(n)}(t)\right]B\right)\\
&+&\dfrac{N-n}{N}\sum\limits_{j=1}^n \text{tr}\left(\text{tr}_{\{n+1\}}[V_{j\, n+1}^{[n+1]}+V_{n+1,\,j}^{[n+1]},\rho_{N}^{(n+1)}(t)]B\right).
\end{eqnarray*}

Let $B \in \cB(\mathbb{H}^{\otimes n})$, and we have \begin{eqnarray}
\nonumber&&\text{tr}\left(\left[\sum\limits_{j=1}^N A_j^{[N]} + \dfrac{1}{N}\sum\limits_{i \neq j;i,j=1}^N V_{ij}^{[N]},\rho_N(t)\right]B \otimes 1^{\otimes (N-n)}\right)
\end{eqnarray}
\begin{eqnarray}
\nonumber&&=\text{tr}\left(\sum\limits_{j=1}^N A_j^{[N]} \rho_N(t)\,\,B\otimes 1^{\otimes (N-n)} + \dfrac{1}{N}\sum\limits_{i \neq j;i,j=1}^N V_{ij}^{[N]}\rho_N(t)\,\, B \otimes 1^{\otimes (N-n)}\right.\\
\nonumber&&\left. - \rho_N(t)\sum\limits_{j=1}^N A_j^{[N]} \,\,B\otimes 1^{\otimes (N-n)} - \dfrac{1}{N}\rho_N(t)\sum\limits_{i \neq j;i,j =1}^N V_{ij}^{[N]}\,\,B \otimes 1^{\otimes (N-n)}\right)\\
\nonumber&&=\text{tr}\left(B\otimes 1^{\otimes (N-n)} \,\, \sum\limits_{j=1}^N A_j^{[N]} \rho_N(t) + \dfrac{1}{N} B\otimes 1^{\otimes (N-n)} \,\, \sum\limits_{i \neq j;i,j = 1}^N V_{ij}^{[N]} \rho_N(t)\right.\\
\nonumber&&\left. - \rho_N(t)\sum\limits_{j=1}^N A_j^{[N]} \,\, B\otimes 1^{\otimes (N-n)} - \dfrac{1}{N}\rho_N(t) \sum\limits_{i \neq j;i,j =1}^N V_{ij}^{[N]}\,\, B\otimes 1^{\otimes (N-n)}\right)\\
\label{As}&&=\text{tr}\left(B\otimes 1^{\otimes (N-n)} \,\, \sum\limits_{j=1}^N A_j^{[N]} \rho_N(t)- \rho_N(t)\sum\limits_{j=1}^N A_j^{[N]} \,\, B\otimes 1^{\otimes (N-n)}\right)\\
\label{Vs}&&+\text{tr}\left(\dfrac{1}{N} B\otimes 1^{\otimes (N-n)} \,\, \sum\limits_{i \neq j;i,j = 1}^N V_{ij}^{[N]} \rho_N(t)- \dfrac{1}{N}\rho_N(t) \sum\limits_{i \neq j;i,j =1}^N V_{ij}^{[N]}\,\, B\otimes 1^{\otimes (N-n)}\right).
\end{eqnarray} Line (\ref{As}) can be rewritten as
\begin{eqnarray}
\label{Asn}&&\text{tr}\left(B\otimes 1^{\otimes (N-n)} \,\, \sum\limits_{j=1}^n A_j^{[N]} \rho_N(t)\right)-\text{tr}\left(\rho_N(t)\sum\limits_{j=1}^n A_j^{[N]} \,\, B\otimes 1^{\otimes (N-n)}\right)\\
\label{AsN}&&+\text{tr}\left(B\otimes 1^{\otimes (N-n)} \,\, \sum\limits_{j=n+1}^N A_j^{[N]}\rho_N(t)\right)-\text{tr}\left(\sum\limits_{j=n+1}^N A_j^{[N]} \,\, B\otimes 1^{\otimes (N-n)}\rho_N(t)\right).
\end{eqnarray}
Notice that $B\otimes 1^{\otimes (N-n)} \,\, \sum\limits_{j=n+1}^N A_j^{[N]} = \sum\limits_{j=n+1}^N A_j^{[N]} \,\, B\otimes I^{\otimes (N-n)}$, and so line (\ref{AsN}) is equal to zero (even without taking the trace into account). Notice that in line (\ref{Asn}), $A_j^{[N]}=A_j^{[n]}\otimes 1^{\otimes (N-n)}$ for $j\leq n$, thus line (\ref{Asn}) can be written as
\begin{eqnarray*}
&&\text{tr}\left(B\sum\limits_{j=1}^n A_j^{[n]}\otimes 1^{\otimes (N-n)} \rho_N(t)\right)-\text{tr}\left(\rho_N(t)\sum\limits_{j=1}^n A_j^{[n]}B\,\,\otimes 1^{\otimes (N-n)}\right)\\
&&=\text{tr}\left(\rho_N^{(n)}(t) B\sum\limits_{j=1}^n A_j^{[n]}\right) - \text{tr}\left(\rho_N^{(n)}(t)\sum\limits_{j=1}^n A_j^{[n]} B\right) \hskip.2in \left(\text{by \cite[equation (2.11)]{Attal}}\right)\\
&&=\text{tr}\left(\left[\sum\limits_{j=1}^n A_j^{[n]},\rho_N^{(n)}(t)\right]B\right).
\end{eqnarray*}

Line (\ref{Vs}) can be rewritten as
\begin{eqnarray}
\nonumber&&\text{tr}\left(\dfrac{1}{N} B\otimes 1^{\otimes (N-n)} \,\, \sum\limits_{i \neq j;i,j = 1}^n V_{ij}^{[N]} \rho_N(t)\right)+\text{tr}\left(\dfrac{1}{N}B\otimes 1^{\otimes (N-n)} \,\, \sum\limits_{1 \leq i \leq n < j \leq N} V_{ij}^{[N]}\rho_N(t)\right)\\
\nonumber&&+\text{tr}\left(\dfrac{1}{N}B\otimes 1^{\otimes (N-n)} \,\, \sum\limits_{1 \leq j \leq n < i \leq N}V_{ij}^{[N]}\rho_N(t)\right)+\text{tr}\left(\dfrac{1}{N}B \otimes 1^{\otimes (N-n)} \,\, \sum\limits_{i \neq j;i,j =n+1}^N V_{ij}^{[N]}\rho_N(t)\right)\\
\nonumber&&-\text{tr}\left(\dfrac{1}{N}\rho_N(t) \sum\limits_{i \neq j;i,j =1}^n V_{ij}^{[N]}\,\, B\otimes 1^{\otimes (N-n)}\right)-\text{tr}\left(\dfrac{1}{N}\rho_N(t) \sum\limits_{1 \leq i \leq n < j \leq N} V_{ij}^{[N]}\,\, B\otimes 1^{\otimes (N-n)}\right)\\
\nonumber&&-\text{tr}\left(\dfrac{1}{N}\rho_N(t) \sum\limits_{1 \leq j \leq n < i \leq N} V_{ij}^{[N]}\,\, B\otimes 1^{\otimes (N-n)}\right)-\text{tr}\left(\dfrac{1}{N}\rho_N(t) \sum\limits_{i \neq j;i,j = n+1}^N V_{ij}^{[N]}\,\, B\otimes 1^{\otimes (N-n)}\right).
\end{eqnarray} 

The first and fifth terms of the above expression give
\begin{eqnarray}
\label{Vsn}&&\text{tr}\left(\dfrac{1}{N} B\otimes 1^{\otimes (N-n)} \,\, \sum\limits_{i \neq j;i,j = 1}^n V_{ij}^{[N]} \rho_N(t)\right)-\text{tr}\left(\dfrac{1}{N}\rho_N(t) \sum\limits_{i \neq j;i,j =1}^n V_{ij}^{[N]}\,\, B\otimes 1^{\otimes (N-n)}\right).
\end{eqnarray} The second and sixth terms of the same expression give
\begin{eqnarray}
\label{Vsij}&&\text{tr}\left(\dfrac{1}{N}B\otimes 1^{\otimes (N-n)} \sum\limits_{1 \leq i \leq n < j \leq N}V_{ij}^{[N]}\rho_N(t)\right)-\text{tr}\left(\dfrac{1}{N}\rho_N(t) \sum\limits_{1 \leq i \leq n < j \leq N} V_{ij}^{[N]} B\otimes 1^{\otimes (N-n)}\right).
\end{eqnarray} The third and seventh terms of the same expression give
\begin{eqnarray}
\label{Vsji}&&\text{tr}\left(\dfrac{1}{N}B\otimes 1^{\otimes (N-n)}  \sum\limits_{1 \leq j \leq n < i \leq N} V_{ij}^{[N]}\rho_N(t)\right)-\text{tr}\left(\dfrac{1}{N}\rho_N(t) \sum\limits_{1 \leq j \leq n < i \leq N} V_{ij}^{[N]} B\otimes 1^{\otimes (N-n)}\right).
\end{eqnarray} The fourth and eighth terms of the same expression give
\begin{eqnarray}
\label{VsN}&&\text{tr}\left(\dfrac{1}{N}B \otimes 1^{\otimes (N-n)} \sum\limits_{i \neq j;i,j =n+1}^N V_{ij}^{[N]}\rho_N(t)\right)-\text{tr}\left(\dfrac{1}{N} \sum\limits_{i \neq j;i,j = n+1}^N V_{ij}^{[N]} B\otimes 1^{\otimes (N-n)}\rho_N(t)\right).
\end{eqnarray}

Notice that $B \otimes I^{\otimes (N-n)}\,\, \sum\limits_{i \neq j;i,j =n+1}^N V_{ij}^{[N]}=\sum\limits_{i \neq j;i,j =n+1}^N V_{ij}^{[N]}\,\,B \otimes I^{\otimes (N-n)}$, and so line (\ref{VsN}) is equal to zero (even without taking the trace into account).

Notice that $V_{ij}^{[N]}=V_{ij}^{[n]}\otimes 1^{\otimes (N-n)}$ for $i,j\leq n$,, thus line (\ref{Vsn}) can be rewritten as
\begin{eqnarray*}
&&\dfrac{1}{N}\text{tr}\left(B\sum\limits_{i\neq j; i,j=1}^n V_{ij}^{[n]}\otimes 1^{\otimes (N-n)} \,\, \rho_N(t)\right)-\dfrac{1}{N}\text{tr}\left(\sum\limits_{i\neq j;i,j=1}^n V_{ij}^{[n]}B\otimes 1^{\otimes (N-n)} \,\, \rho_N(t)\right)\\
&&=\dfrac{1}{N}\text{tr}\left(\rho_N^{(n)}(t)B\sum\limits_{i\neq j;i,j=1}^n V_{ij}^{[n]}\right)-\dfrac{1}{N}\text{tr}\left(\rho_N^{(n)}(t)\sum\limits_{i\neq j;i,j=1}^n V_{ij}^{[n]}B\right) \hskip.2in \left(\text{by \cite[equation (2.11)]{Attal}}\right)\\
&&= \text{tr}\left(\left[\dfrac{1}{N}\sum\limits_{i \neq j =1}^n V_{ij}^{[n]},\rho_N^{(n)}(t)\right]B\right).
\end{eqnarray*}

There are $N-n$ values of $j$ in line (\ref{Vsij}) and by symmetry of $\rho_N(t)$ we can replace all of these values of $j$ by $n+1$ and thus we have that line (\ref{Vsij}) can be rewritten as
\begin{eqnarray*}
&&\dfrac{N-n}{N} \text{tr}\left(B\otimes 1^{\otimes (N-n)} \,\, \sum\limits_{i=1}^n V_{i\,n+1}^{[N]}\rho_N(t)\right)-\dfrac{N-n}{N} \text{tr}\left(\sum\limits_{i=1}^n V_{i\,n+1}^{[N]}\,\,B\otimes 1^{\otimes (N-n)} \,\, \rho_N(t)\right).
\end{eqnarray*} Notice that $V_{i\, n+1}^{[N]}=V_{i\, n+1}^{[n+1]}\otimes 1^{\otimes (N-(n+1))}$ for $i\leq n$, thus the last displayed equation is equal to
\begin{eqnarray*}
&&\dfrac{N-n}{N} \text{tr}\left(\left(B\otimes 1\,\,\sum\limits_{i=1}^n V_{i\,n+1}^{[n+1]}\right)\otimes 1^{\otimes (N-n-1)} \,\, \rho_N(t)\right)\\
&&-\dfrac{N-n}{N} \text{tr}\left(\left(\sum\limits_{i=1}^n V_{i\,n+1}^{[n+1]}\,\,B\otimes 1\right)\,\,\otimes 1^{\otimes (N-n-1)} \,\, \rho_N(t)\right)
\end{eqnarray*} and therefore by \cite[equation (2.11)]{Attal} the last displayed expression is equal to
\begin{eqnarray*}
&&\dfrac{N-n}{N} \text{tr}\left(B\otimes 1 \,\, \sum\limits_{i=1}^n V_{i\,n+1}^{[n+1]} \,\, \rho_N^{(n+1)}(t)\right)-\dfrac{N-n}{N} \text{tr}\left(\sum\limits_{i=1}^n V_{i\,n+1}^{(n+1)}\,\, B\otimes 1 \,\, \rho_N^{(n+1)}(t)\right)\\
&&=\dfrac{N-n}{N}\sum\limits_{j=1}^n \text{tr}\left([V_{j\,n+1}^{[n+1]},\rho_{N}^{(n+1)}(t)]B\otimes 1\right)=\dfrac{N-n}{N}\sum\limits_{j=1}^n \text{tr}\left(\text{tr}_{\{n+1\}}[V_{j\,n+1}^{[n+1]},\rho_{N}^{(n+1)}(t)]B\right),
\end{eqnarray*} where again we used \cite[equation (2.11)]{Attal} to obtain the last equality.

Similarly, line (\ref{Vsji}) can be rewritten as
\begin{eqnarray*}
\dfrac{N-n}{N}\sum\limits_{j=1}^n \text{tr}\left(\text{tr}_{\{n+1\}}[V_{n+1\,j}^{[n+1]},\rho_{N}^{(n+1)}(t)]B\right).
\end{eqnarray*}

Thus, (\ref{DiffEq1}) and (\ref{claim1}) lead to the result
\begin{eqnarray*}
i\dfrac{d}{dt}\rho_N^{(n)}(t)=[H_{n,N},\rho_N^{(n)}(t)]+\dfrac{N-n}{N}\sum\limits_{j=1}^n \text{tr}_{\{n+1\}}[V_{j \,n+1}^{[n+1]}+V_{n+1\, j}^{[n+1]},\rho_{N}^{(n+1)}(t)].
\end{eqnarray*}
\end{proof}

The next proposition concludes with a family of differential equations which is satisfied by $\left(\rho(t)^{\otimes n}\right)_{n=1}^\infty$. This family of differential equations is similar to the ones displayed in equation (\ref{rhoNDiffEq}).

\begin{Prop}\label{Hartree2}
Let $\rho(0) \in \mathcal{D}(\mathbb{H})$. If $\rho(t)$ is the solution to the differential equation
\begin{eqnarray}
i\dfrac{d}{dt}\rho(t)=[A,\rho(t)]+\text{tr}_{\{2\}}\left[V_{12}^{[2]}+V_{21}^{[2]},\rho(t)\otimes\rho(t)\right]\label{HartreeEq}
\end{eqnarray} for $t\geq 0$ (which is called the Hartree Equation), with initial condition $\rho(0)$, (see Remark~\ref{BoveDaPratoFanoConditions}),
then we have that $\left(\rho(t)^{\otimes n}\right)_{n=1}^\infty$ satisfies the family of differential equations
\begin{eqnarray}\label{rhoDiffEq}
i\dfrac{d}{dt}\rho(t)^{\otimes n} = \sum\limits_{j=1}^n\left[A_j^{[n]},\rho(t)^{\otimes n}\right] + \sum\limits_{j=1}^n \text{tr}_{\{n+1\}}\left[V_{j \, n+1}^{[n+1]}+V_{n+1 \, j}^{[n+1]}, \rho(t)^{\otimes (n+1)}\right].
\end{eqnarray}
\end{Prop}

\begin{proof}
We have
\begin{eqnarray}
\nonumber&&i\dfrac{d}{dt}\rho(t)^{\otimes n}=\sum\limits_{j=1}^n \rho(t)^{\otimes (j-1)}\otimes i\dfrac{d}{dt}\rho(t)\otimes \rho(t)^{\otimes (n-j)} \hskip.2in (\text{``product'' rule})\\
\nonumber&=&\sum\limits_{j=1}^n \rho(t)^{\otimes (j-1)}\otimes  \left([A,\rho(t)]+\text{tr}_{\{2\}}\left[V_{12}^{[2]}+V_{21}^{[2]},\rho(t)\otimes\rho(t)\right]\right)\otimes \rho(t)^{\otimes (n-j)}
\end{eqnarray} by assumption. The last expression splits into the following two parts
\begin{eqnarray}
\label{CommA}&&\sum\limits_{j=1}^n \rho(t)^{\otimes (j-1)}\otimes [A,\rho(t)]\otimes \rho(t)^{\otimes (n-j)}\\
\label{CommV}&+&\sum\limits_{j=1}^n \rho(t)^{\otimes (j-1)}\otimes  \text{tr}_{\{2\}}\left[V_{12}^{[2]}+V_{21}^{[2]},\rho(t)\otimes\rho(t)\right]\otimes \rho(t)^{\otimes (n-j)}
\end{eqnarray}

Line (\ref{CommA}) can be rewritten as
\begin{eqnarray}
\nonumber&&\left[A,\rho(t)\right]\otimes \rho(t)^{\otimes(n-1)}+\rho(t)\otimes\left[A,\rho(t)\right]\otimes\rho(t)^{\otimes(n-1)}+\cdots + \rho(t)^{\otimes (n-1)}\otimes \left[A,\rho(t)\right]\\
\label{part1}&=& \sum\limits_{j=1}^n\left[A_j^{[n]},\rho(t)^{\otimes n}\right].
\end{eqnarray}

We claim that for $j\leq n$,
\begin{eqnarray}
\label{eq1} \text{tr}_{\{j+1\}}\left(\left(V_{j\,j+1}^{[n+1]}+V_{j+1\,j}^{[n+1]}\right)\rho(t)^{\otimes (n+1)}\right)=\text{tr}_{\{n+1\}}\left(\left(V_{j\,n+1}^{[n+1]}+V_{n+1\,j}^{[n+1]}\right)\rho(t)^{\otimes (n+1)}\right).
\end{eqnarray} Indeed, by \cite[equation (2.11)]{Attal}, for any $B_1,...,B_n \in \cB(\mathbb{H})$, we have by the symmetry of $\rho(t)^{\otimes (n+1)}$,
\begin{eqnarray*}
&&\text{tr}\left(\text{tr}_{\{j+1\}}\left(\left(V_{j\,j+1}^{[n+1]}+V_{j+1\,j}^{[n+1]}\right)\rho(t)^{\otimes (n+1)}\right)B_1\otimes \cdots \otimes B_n\right)\\
&=&\text{tr}\left(\left(V_{j\,j+1}^{[n+1]}+V_{j+1\, j}^{[n+1]}\right)\rho(t)^{\otimes (n+1)}B_1\otimes \cdots \otimes B_j \otimes 1 \otimes B_{j+1} \otimes \cdots \otimes B_n\right)\\
&=&\text{tr}\left(B_1\otimes \cdots \otimes B_j \otimes 1 \otimes B_{j+1} \otimes \cdots\otimes B_n\left(V_{j\,j+1}^{[n+1]}+V_{j+1\, j}^{[n+1]}\right)\rho(t)^{\otimes (n+1)}\right)\\
&=&\text{tr}\left(B_1\otimes \cdots \otimes B_n \otimes 1\left(V_{j\,n+1}^{[n+1]}+V_{n+1\, j}^{[n+1]}\right)\rho(t)^{\otimes (n+1)}\right)\\
&=&\text{tr}\left(\text{tr}_{\{n+1\}}\left(\left(V_{j\,n+1}^{[n+1]}+V_{n+1\,j}^{[n+1]}\right)\rho(t)^{\otimes (n+1)}\right)B_1\otimes \cdots \otimes B_n\right)
\end{eqnarray*} where for the second to last equality we used the symmetry of $\rho(t)^{\otimes (n+1)}$ to move each $B_k$ (for $k\geq j+1$) to the $k$th spot and $1$ to the $n+1$st spot.

Similar to equation (\ref{eq1}), we have that for $j \leq n$,
\begin{eqnarray}
\label{eq2}\text{tr}_{\{j+1\}}\left(\rho(t)^{\otimes (n+1)}\left(V_{j\,j+1}^{[n+1]}+V_{j+1\, j}^{[n+1]}\right)\right)=\text{tr}_{\{n+1\}}\left(\rho(t)^{\otimes (n+1)}\left(V_{j\,n+1}^{[n+1]}+V_{n+1\,j}^{[n+1]}\right)\right).
\end{eqnarray}
Line (\ref{CommV}) can be rewritten as
\begin{eqnarray}
\nonumber&&\sum\limits_{j=1}^n \rho(t)^{\otimes (j-1)}\otimes  \text{tr}_{\{2\}}\left(\left(V_{12}^{[2]}+V_{21}^{[2]}\right)\rho(t)\otimes\rho(t)\right)\otimes \rho(t)^{\otimes (n-j)}\\
\nonumber&-&\sum\limits_{j=1}^n \rho(t)^{\otimes (j-1)}\otimes  \text{tr}_{\{2\}}\left(\rho(t)\otimes\rho(t)\left(V_{12}^{[2]}+V_{21}^{[2]}\right)\right)\otimes \rho(t)^{\otimes (n-j)}\\
\label{eq3}&=& \sum\limits_{j=1}^n \text{tr}_{\{j+1\}}\left(\left(V_{j\,j+1}^{[n+1]}+V_{j+1\,j}^{[n+1]}\right)\rho(t)^{\otimes (n+1)}-\rho(t)^{\otimes (n+1)}\left(V_{j\,j+1}^{[n+1]}+V_{j+1\,j}^{[n+1]}\right)\right)
\end{eqnarray}. 

By equations (\ref{eq1}) and (\ref{eq2}), line (\ref{eq3}) is equal to
\begin{eqnarray}
\nonumber&&\sum\limits_{j=1}^n\left(\text{tr}_{\{n+1\}}\left(\left(V_{j\,n+1}^{[n+1]}+V_{n+1\,j}^{[n+1]}\right)\rho(t)^{\otimes (n+1)}\right)\right.\\
\nonumber&-&\left.\text{tr}_{\{n+1\}}\left(\rho(t)^{\otimes (n+1)}\left(V_{j\,n+1}^{[n+1]}+V_{n+1\,j}^{[n+1]}\right)\right)\right)\\
\label{part2}&=&\sum\limits_{j=1}^n \text{tr}_{\{n+1\}}\left[V_{j \, n+1}^{[n+1]}+V_{n+1 \, j}^{[n+1]}, \rho(t)^{\otimes (n+1)}\right].
\end{eqnarray}
Of course (\ref{part1}) and (\ref{part2}) complete the proof.
\end{proof}

Equation (\ref{HartreeEq}) has a unique solution. This solution $\rho(t)$ is self-adjoint trace class for each $t\geq 0$, see \cite[Theorem 4.1]{BoveDaPratoFano}. We will see in the proof of Theorem~\ref{PropResult} that $\rho(t)$ is a density operator for each $t\geq 0$. The next remark checks that equation~(\ref{HartreeEq}) satisfies the conditions of \cite[Theorem 4.1]{BoveDaPratoFano}.

\begin{Rmk}\label{BoveDaPratoFanoConditions}
Let $X$ be the real Banach space of self-adjoint trace class operators on $\mathbb{H}$. Define the mapping $f:X\rightarrow X$ by $f(T)=-i\text{tr}_{\{2\}}\left[V_{12}^{[2]}+V_{21}^{[2]},T\otimes T\right]$. 

We will first prove that $f$ is locally Lipschitzian. Let $r\geq 0$, and $T,S\in X$ with $\left|\left|T\right|\right|_1 \leq r$ and $\left|\left|S\right|\right|_1 \leq r$. Then
\begin{eqnarray*}
\left|\left|f(T)-f(S)\right|\right|_1 &=& \left|\left|-i\text{tr}_{\{2\}}\left[V_{12}^{[2]}+V_{21}^{[2]},T\otimes T\right]+i\text{tr}_{\{2\}}\left[V_{12}^{[2]}+V_{21}^{[2]},S\otimes S\right]\right|\right|_1\\
&=&\left|\left|\text{tr}_{\{2\}}\left[V_{12}^{[2]}+V_{21}^{[2]},T\otimes T - S\otimes S\right]\right|\right|_1\\
&\leq& \left|\left|\left[V_{12}^{[2]}+V_{21}^{[2]},T\otimes T - S\otimes S\right]\right|\right|_1 \leq 4\left|\left|V\right|\right|_\infty \left|\left|T\otimes T - S\otimes S\right|\right|_1\\
&=& 4\left|\left|V\right|\right|_\infty \left|\left|T\otimes T - T\otimes S + T\otimes S - S\otimes S\right|\right|_1\\
&\leq& 4\left|\left|V\right|\right|_\infty \left(\left|\left|T\otimes\left(T-S\right)\right|\right|_1+\left|\left|\left(T-S\right)\otimes S\right|\right|_1\right)\\
&\leq& 4r\left|\left|V\right|\right|_\infty \left|\left|T-S\right|\right|_1
\end{eqnarray*} which shows that $f$ is locally Lipschitzian.

We will now prove that $f$ satisfies $||T||_1\leq ||T-\alpha f(T)||_1$ for all $\alpha\geq 0$ and $T\in X$ \cite[inequality (4.1)]{BoveDaPratoFano}. Let $\alpha \geq 0$, $T\in X$, and $S=T-\alpha f(T)$. We have that $T=\sum\limits_{i=1}^\infty \lambda_i \left.\left|x_i\right>\right.\left.\left<x_i\right|\right.$ for some orthonormal sequence $\left(x_i\right)_{i=1}^\infty \subset \mathbb{H}$ and some sequence $\left(\lambda_i\right)_{i=1}^\infty \subset \mathbb{R}$. Define $\sigma = \sum\limits_{i=1}^\infty \text{sign}(\lambda_i)\left.\left|x_i\right>\right.\left.\left<x_i\right|\right.$. Thus, $\left|T\right|=T\sigma = \sigma T$, and
\begin{eqnarray*}
\text{tr}\left(f(T)\sigma\right)&=&-i\text{tr}\left(\text{tr}_{\{2\}}\left[V_{12}^{[2]}+V_{21}^{[2]},T\otimes T\right]\sigma\right)=-i\text{tr}\left(\left[V_{12}^{[2]}+V_{21}^{[2]},T\otimes T\right]\sigma\otimes 1\right)\\
&=&-i\text{tr}\left(\left(V_{12}^{[2]}+V_{21}^{[2]}\right)\left(T\otimes T\right)\left(\sigma\otimes 1\right) - \left(T\otimes T\right)\left(V_{12}^{[2]}+V_{21}^{[2]}\right)\left(\sigma\otimes 1\right)\right)\\
&=& -i\text{tr}\left(\left(V_{12}^{[2]}+V_{21}^{[2]}\right)\left|T\right|\otimes T - \left|T\right|\otimes T \left(V_{12}^{[2]}+V_{21}^{[2]}\right)\right)\\
&=& -i\text{tr}\left(\left[V_{12}^{[2]}+V_{21}^{[2]},\left|T\right|\otimes T\right]\right)=0.
\end{eqnarray*} By the cyclicity of the trace, we also have $\text{tr}\left(\sigma f(T)\right)=0$. Hence, we have
\begin{eqnarray*}
\left|\left|T\right|\right|_1 &=& \dfrac{1}{2}\text{tr}\left(T\sigma + \sigma T\right) = \dfrac{1}{2}\text{tr}\left(\left(S+\alpha f(T)\right)\sigma + \sigma\left(S+\alpha f(T)\right)\right)\\
&=& \dfrac{1}{2}\text{tr}\left(S\sigma + \sigma S\right)+\dfrac{\alpha}{2}\text{tr}\left(f(T)\sigma + \sigma f(T)\right) = \dfrac{1}{2}\text{tr}\left(S\sigma + \sigma S\right)\\
&\leq& \dfrac{1}{2}\text{tr}\left|S\sigma + \sigma S\right| \leq \dfrac{1}{2}\left(\left|\left|S\right|\right|_1 \left|\left|\sigma\right|\right|_\infty + \left|\left|\sigma\right|\right|_\infty \left|\left|S\right|\right|_1\right) = \left|\left|\sigma\right|\right|_\infty \left|\left|S\right|\right|_1\\
&=& \left|\left|S\right|\right|_1 = \left|\left|T-\alpha f(T)\right|\right|_1.
\end{eqnarray*}

Finally, it is clear that $M(\cdot)=-i\left[A,\cdot\right]$ defines the infinitesimal generator of the contraction semigroup $T_t(\cdot)=e^{-itA}\left(\cdot\right)e^{itA}$ in $X$. Thus, all of the conditions of \cite[Theorem 4.1]{BoveDaPratoFano} are satisfied.
\end{Rmk}

The similarity of the two equations (\ref{rhoNDiffEq}) and (\ref{rhoDiffEq}) helps to prove the propagation of chaos presented in the following theorem. The idea of the proof of this theorem comes from Ducomet~\cite[Theorem 3.1]{Ducomet}, and Bardos, Golse, Gottlieb, and Mauser~\cite[Theorem 3.1]{Bardos}.

\begin{Thm}\label{PropResult}
Let a sequence $(\rho_N(0))_{N=1}^\infty$ of density operators be $\rho(0)$-chaotic where $\rho(0) \in \mathcal{D}(\mathbb{H})$. Let $(H_N)_{N=1}^\infty$ be a sequence of Hamiltonians defined by equation (\ref{HamiltonianSeq}). Then, for each fixed $t \geq 0$, the sequence of density operators $(\rho_N(t))_{N=1}^\infty$ defined in equation (\ref{NewDensity}) is $\rho(t)$-chaotic where $\rho(t)$ is the solution of the Hartree equation (equation (\ref{HartreeEq})) with initial condition $\rho(0)$. Thus chaos propagates with respect to the Hamiltonians $(H_N)_{N=1}^\infty$.
\end{Thm}

Bardos, Golse, and Mauser \cite[Theorem 5.4]{BardosGolseMauser} prove propagation of chaos for a similar sequence of Hamiltonians using a different approach. In the present paper, we consider density operators while the above authors present their results in terms of the equivalent density matrices. There are several differences between Theorem~\ref{PropResult} and Theorem $5.4$ of \cite{BardosGolseMauser}. For example, the sequence of Hamiltonians in Theorem~\ref{PropResult} is constructed using an arbitrary bounded self-adjoint operator $A$ on an arbitrary Hilbert space, whereas in \cite[Theorem 5.4]{BardosGolseMauser}, the authors consider $-\Delta$ on an appropriate Hilbert space. In addition, the sequence of Hamiltonians present in Theorem~\ref{PropResult} is constructed using an arbitrary bounded self-adjoint operator $V$, whereas in \cite[Theorem 5.4]{BardosGolseMauser}, the authors consider a bounded function $V$ on $\mathbb{R}_{+}$ which acts multiplicatively on the wave function, is bounded from below, and vanishes at infinity. Another paper that examines an evolution via a similar sequence of Hamiltonians using the language of wave functions and density matrices is \cite{ChenLeeLee}.

\begin{proof}
Fix $T_0\in\mathbb{N}\cup\{0\}$. Let $K_{T_0}:=\sup\left\{\left|\left|\rho(t)\right|\right|_1:t\in\left[T_0,T_0+1\right]\right\}<\infty$ since by \cite[Theorem 4.1]{BoveDaPratoFano} we know that the self-adjoint trace class solution $\rho(t)$ to the Hartree equation is continuous with respect to time. Let $\delta_{T_0}:=\dfrac{1}{\left[8\left(\max\{||V||_\infty,1\}\right)K_{T_0}\right]+1}$ where $\left[8\left(\max\{||V||_\infty,1\}\right)K_{T_0}\right]$ is the integer part of $8\left(\max\{||V||_\infty,1\}\right)K_{T_0}$.

In order to prove Theorem~\ref{PropResult} we will show the following: 

Fix $t_0:=T_0+k\delta_{T_0}$ for some $k\in\{0,1,...,\left[8\left(\max\{||V||_\infty,1\}\right)K_{T_0}\right]\}$. Assume that $\left(\rho_N(t_0)\right)_{N=1}^\infty$ is $\rho(t_0)$ chaotic where $\rho(t_0)\in\mathcal{D}(\mathbb{H})$. Then for $t\in[t_0,t_0+\delta_{T_0}]$, $\left(\rho_N(t)\right)_{N=1}^\infty$ is $\rho(t)$ chaotic where $\rho(t)\in\mathcal{D}(\mathbb{H})$ is the solution to the Hartree equation (equation~(\ref{HartreeEq})) with initial condition $\rho(t_0)$.

For $t \in [t_0,\infty)$, by Proposition~\ref{Hartree2}, for each $n \in \mathbb{N}$,
\begin{eqnarray}
\label{diff2} i\dfrac{d}{dt}\rho(t)^{\otimes n} = \sum\limits_{j=1}^n\left[A_j^{[n]},\rho(t)^{\otimes n}\right] + \sum\limits_{j=1}^n \text{tr}_{\{n+1\}}\left[V_{j \, n+1}^{[n+1]}+V_{n+1 \, j}^{[n+1]}, \rho(t)^{\otimes (n+1)}\right].
\end{eqnarray}

Also notice that for each $n,N \in \mathbb{N}$ with $n \leq N-1$, by Proposition~\ref{Hartree1}
\begin{eqnarray}
\nonumber i\dfrac{d}{dt}\rho_N^{(n)}(t)&=&[H_{n,N},\rho_N^{(n)}(t)]+\dfrac{N-n}{N}\sum\limits_{j=1}^n \text{tr}_{\{n+1\}}[V_{j \,n+1}^{[n+1]}+V_{n+1\, j}^{[n+1]},\rho_{N}^{(n+1)}(t)]\\
\nonumber&=& \sum\limits_{j=1}^n \left[A_j^{[n]},\rho_N^{(n)}(t)\right]+\sum\limits_{j=1}^n \text{tr}_{\{n+1\}}[V_{j \,n+1}^{[n+1]}+V_{n+1\, j}^{[n+1]},\rho_{N}^{(n+1)}(t)]\\
\nonumber&+&\dfrac{1}{N}\sum\limits_{i\neq j=1}^n \left[V_{i\,j}^{[n]},\rho_N^{(n)}(t)\right]-\dfrac{n}{N}\sum\limits_{j=1}^n \text{tr}_{\{n+1\}}[V_{j \,n+1}^{[n+1]}+V_{n+1\, j}^{[n+1]},\rho_{N}^{(n+1)}(t)]\\
\label{diff1}&=&\mathcal{L}_n(\rho_N^{(n)}(t))+\sum\limits_{j=1}^n \text{tr}_{\{n+1\}}[V_{j \,n+1}^{[n+1]}+V_{n+1\, j}^{[n+1]},\rho_{N}^{(n+1)}(t)]+\epsilon_n(t,N,\rho_N(t_0))
\end{eqnarray} where $\mathcal{L}_n(\cdot):=\sum\limits_{j=1}^n \left[A_j^{[n]},\cdot\right]$ and $$\epsilon_n(t,N,\rho_N(t_0)):=\dfrac{1}{N}\sum\limits_{i\neq j=1}^n \left[V_{i\,j}^{[n]},\rho_N^{(n)}(t)\right]-\dfrac{n}{N}\sum\limits_{j=1}^n \text{tr}_{\{n+1\}}[V_{j \,n+1}^{[n+1]}+V_{n+1\, j}^{[n+1]},\rho_{N}^{(n+1)}(t)].$$

Define $E_{n,N}(t):=\rho_N^{(n)}(t)-\rho(t)^{\otimes n}$ for each $n \leq N$. Then, by subtracting (\ref{diff2}) from (\ref{diff1}), we obtain that for each $n,N \in \mathbb{N}$ with $n \leq N-1$,
\begin{eqnarray*}
i\dfrac{d}{dt}E_{n,N}(t) = \sum\limits_{j=1}^n\left[A_j^{[n]},E_{n,N}(t)\right] + \sum\limits_{j=1}^n \text{tr}_{\{n+1\}}\left[V_{j \, n+1}^{[n+1]}+V_{n+1 \, j}^{[n+1]}, E_{n+1,N}(t)\right]+\epsilon_n(t,N,\rho_N(t_0)).
\end{eqnarray*}

The next step is to obtain an upper bound for the trace class norm of $E_{n,N}(t)$ which does not involve $||A||_\infty$. In order to do this, we define a new evolution $\mathcal{U}_{n,t}$ which will be evaluated at $E_{n,N}(t)$. 

We define $\mathcal{U}_{n,t}(\cdot):=e^{it\mathcal{L}_n}(\cdot)=e^{it\sum\limits_{j=1}^n A_j^{[n]}}(\cdot)e^{-it\sum\limits_{j=1}^n A_j^{[n]}}$. We claim that $\mathcal{U}_{n,t}$ is an isometry on the trace class operators on $\mathbb{H}^{\otimes n}$ for each $n \in \mathbb{N}$ and $t\in[0,\infty)$. Indeed, if $T \in \cB(\mathbb{H}^{\otimes n})$ is a trace class operator, then
\begin{eqnarray*}
||\mathcal{U}_{n,t}(T)||_1 = ||e^{it\sum\limits_{j=1}^n A_j^{[n]}}Te^{-it\sum\limits_{j=1}^n A_j^{[n]}}||_1 \leq ||e^{it\sum\limits_{j=1}^n A_j^{[n]}}||_\infty ||T||_1 ||e^{-it\sum\limits_{j=1}^n A_j^{[n]}}||_\infty = ||T||_1,
\end{eqnarray*} and similarly, by observing that $T=e^{-it\sum\limits_{j=1}^n A_j^{[n]}}\mathcal{U}_{n,t}(T)e^{it\sum\limits_{j=1}^n A_j^{[n]}}$, we get the reverse inequality.

Now define $Z_{n,N}(t):=\mathcal{U}_{n,t}(E_{n,N}(t))$. Then
\begin{eqnarray}
\nonumber&&\dfrac{d}{dt}Z_{n,N}(t)=i\sum\limits_{j=1}^n A_j^{[n]} \mathcal{U}_{n,t}(E_{n,N}(t))-i\mathcal{U}_{n,t}(E_{n,N}(t))\sum\limits_{j=1}^n A_j^{[n]}\\
\nonumber&-&i\mathcal{U}_{n,t}\left(\mathcal{L}_n(E_{n,N}(t))+\sum\limits_{j=1}^n\text{tr}_{\{n+1\}}\left[V_{j\, n+1}^{[n+1]}+V_{n+1\, j}^{[n+1]}, E_{n+1,N}(t)\right]+\epsilon_n(t,N,\rho_N(t_0))\right)\\
\label{diffZ}&=& -i\sum\limits_{j=1}^n \mathcal{U}_{n,t}\left(\text{tr}_{\{n+1\}}\left[V_{j\,n+1}^{[n+1]}+V_{n+1\,j}^{[n+1]}, E_{n+1,N}(t)\right]\right)-i\mathcal{U}_{n,t}\left(\epsilon_n(t,N,\rho_N(t_0))\right)
\end{eqnarray}
where the last equality follows because $\mathcal{U}_{n,t}$ and $\mathcal{L}_n$ commute hence $$i\sum\limits_{j=1}^n A_j^{[n]} \mathcal{U}_{n,t}(E_{n,N}(t))-i\mathcal{U}_{n,t}(E_{n,N}(t))\sum\limits_{j=1}^n A_j^{[n]}-i\mathcal{U}_{n,t}(\mathcal{L}_n(E_{n,N}(t)))=0.$$

Line~(\ref{diffZ}) is continuous from $[0,\infty)$ to the trace class operators on $\mathbb{H}^{\otimes n}$. Thus, the Fundamental Theorem of Calculus is valid \cite[Theorem 2.3]{Mikusinski}. By integrating both sides of equation (\ref{diffZ}), we obtain that, for each $n,N \in \mathbb{N}$ with $n \leq N-1$, 
\begin{eqnarray*}
Z_{n,N}(t)&=&Z_{n,N}(t_0)-i\sum\limits_{j=1}^n\displaystyle\int_{t_0}^t\mathcal{U}_{n,s}\left(\text{tr}_{\{n+1\}} \left[V_{j\,n+1}^{[n+1]}+V_{n+1\,j}^{[n+1]},E_{n+1,N}(s)\right]\right)ds\\
&-& i \displaystyle\int_{t_0}^t \mathcal{U}_{n,s}\left(\epsilon_n(s,N,\rho_N(t_0))\right)ds.
\end{eqnarray*}

We will aim to show that $\lim\limits_{N\rightarrow\infty}||E_{n,N}(t)||_1=0$. We have
\begin{eqnarray}
\nonumber&&||E_{n,N}(t)||_1 = ||Z_{n,N}(t)||_1\leq ||Z_{n,N}(t_0)||_1\\ 
\nonumber&+& \left|\left|\sum\limits_{j=1}^n\int_{t_0}^t\mathcal{U}_{n,s}\left(\text{tr}_{\{n+1\}}\left[V_{j\,n+1}^{[n+1]}+V_{n+1\,j}^{[n+1]},E_{n+1,N}(s)\right]\right)ds\right|\right|_1\\
\nonumber&+& \left|\left|\int_{t_0}^t \mathcal{U}_{n,s}\left(\epsilon_n(s,N,\rho_N(t_0))\right)\right|\right|_1\\
\label{Einq}&\leq& \left|\left|E_{n,N}(t_0)\right|\right|_1+(t-t_0)||\epsilon_n(s,N,\rho_N(t_0))||_1 + 4||V||_\infty \sum\limits_{j=1}^n\int_{t_0}^t \left|\left|E_{n+1,N}(s)\right|\right|_1 ds.
\end{eqnarray} We notice that for every $n,N \in \N$ with $n \leq N-1$ and $s \in [0,\infty)$,
\begin{eqnarray}
\nonumber&&||\epsilon_n(s,N,\rho_N(t_0))||_1\\
\nonumber&=& \left|\left|\dfrac{1}{N}\sum\limits_{i\neq j=1}^n \left[V_{i\,j}^{[n]},\rho_N^{(n)}(s)\right]-\dfrac{n}{N}\sum\limits_{j=1}^n \text{tr}_{\{n+1\}}[V_{j \,n+1}^{[n+1]}+V_{n+1\, j}^{[n+1]},\rho_{N}^{(n+1)}(s)]\right|\right|_1\\
\nonumber&\leq& \dfrac{1}{N}\left|\left|\sum\limits_{i\neq j;i,j=1}^n \left[V_{i\,j}^{[n]},\rho_N^{(n)}(s)\right]\right|\right|_1+\dfrac{n}{N}\left|\left|\sum\limits_{j=1}^n \text{tr}_{\{n+1\}}[V_{j \,n+1}^{[n+1]}+V_{n+1\, j}^{[n+1]},\rho_{N}^{(n+1)}(s)]\right|\right|_1\\
\label{einq}&\leq& \dfrac{n(n-1)}{N}||V||_\infty + \dfrac{4n^2}{N}||V||_\infty \leq \dfrac{5n^2}{N}||V||_\infty.
\end{eqnarray}

Inequalities (\ref{Einq}) and (\ref{einq}) give
\begin{eqnarray}\label{inequality1}
||E_{n,N}(t)||_1 \leq ||E_{n,N}(t_0)||_1 + \dfrac{5n^2}{N}||V||_\infty (t-t_0) + 4||V||_\infty \sum\limits_{j=1}^n \int_{t_0}^t ||E_{n+1,N}(s)||_1 ds.
\end{eqnarray}

Fixing $n \in \mathbb{N}$ and iterating this inequality $m$ more times for $m \in \mathbb{N}$ and $m \leq N-n-1$, we obtain
\begin{eqnarray}
\label{Trick2}&&\hskip.2in\left|\left|E_{n,N}(t)\right|\right|_1 \leq \left|\left|E_{n,N}(t_0)\right|\right|_1+\dfrac{5n^2}{N}||V||_\infty (t-t_0)\\
\nonumber&+&\sum\limits_{k=1}^m (4||V||_\infty)^{k}\left[\sum\limits_{j_1=1}^{n}\sum\limits_{j_2=1}^{n+1}\cdots \sum\limits_{j_k=1}^{n+k-1}\left(\dfrac{(t-t_0)^k}{k!}\left|\left|E_{n+k,N}(t_0)\right|\right|_1 + \dfrac{5(n+k)^2}{N}||V||_\infty\dfrac{(t-t_0)^{k+1}}{(k+1)!}\right)\right]\\
\nonumber&+& (4||V||_\infty)^{m+1}\sum\limits_{j_1=1}^{n}\sum\limits_{j_2=1}^{n+1}\cdots\sum\limits_{j_{m+1}=1}^{n+m} \displaystyle\int_{t_0}^t \displaystyle\int_{t_0}^{t_1} \cdots \displaystyle\int_{t_0}^{t_m}\left|\left|E_{n+m+1,N}(t_{m+1})\right|\right|_1dt_{m+1}dt_m \cdots dt_1.
\end{eqnarray}

Since $\rho_N^{(n+m+1)}(t)$ is a density operator, by the definition of $K_{T_0}$ we have that $$\left|\left|E_{n+m+1,N}(s)\right|\right|_1\leq 1+K_{T_0}^{n+m+1}\leq 2K_{T_0}^{n+m+1}$$ for $s\in[T_0,T_0+1]$ (since $K_{T_0}\geq\left|\left|\rho(T_0)\right|\right|_1=1$). Thus the last line of inequality (\ref{Trick2}) can be bounded above by
\begin{eqnarray*}
&&(4||V||_\infty)^{m+1}\sum\limits_{j_1=1}^{n}\sum\limits_{j_2=1}^{n+1}\cdots\sum\limits_{j_{m+1}=1}^{n+m} \displaystyle\int_{t_0}^t \displaystyle\int_{t_0}^{t_1} \cdots \displaystyle\int_{t_0}^{t_m}2K_{T_0}^{n+m+1} dt_{m+1}dt_m \cdots dt_1\\
&=& 2K_{T_0}^{n+m+1}(4||V||_\infty)^{m+1} n(n+1)\cdots (n+m) \dfrac{(t-t_0)^{m+1}}{(m+1)!}\\
&=& 2K_{T_0}^n\dfrac{n(n+1)\cdots(n+m)}{(m+1)!}(4||V||_\infty K_{T_0} (t-t_0))^{m+1}\\
&=& 2K_{T_0}^n{n+m\choose n-1}(4||V||_\infty K_{T_0} (t-t_0))^{m+1} \leq \dfrac{2K_{T_0}^n}{n!}(n+m)^n(4||V||_\infty K_{T_0} (t-t_0))^{m+1}
\end{eqnarray*} where we used that
\begin{eqnarray*}
{n+m \choose n-1} = \dfrac{(n+m)!}{(n-1)!(m+1)!} \leq \dfrac{(n+m)^{n-1}}{(n-1)!}=\dfrac{n(n+m)^{n-1}}{n!}\leq \dfrac{(n+m)^n}{n!}.
\end{eqnarray*}

Thus, by (\ref{Trick2}), we obtain
\begin{eqnarray*}
\nonumber &&\left|\left|E_{n,N}(t)\right|\right|_1 \leq \left|\left|E_{n,N}(t_0)\right|\right|_1+\dfrac{5n^2}{N}||V||_\infty (t-t_0)\\
\nonumber&+&\sum\limits_{k=1}^m (4||V||_\infty)^{k}\left[\sum\limits_{j_1=1}^{n}\sum\limits_{j_2=1}^{n+1}\cdots \sum\limits_{j_k=1}^{n+k-1}\left(\dfrac{(t-t_0)^k}{k!}\left|\left|E_{n+k,N}(t_0)\right|\right|_1 + \dfrac{5(n+k)^2}{N}||V||_\infty\dfrac{(t-t_0)^{k+1}}{(k+1)!}\right)\right]\\\nonumber
&+& \dfrac{2K_{T_0}^n}{n!}(n+m)^n(4||V||_\infty K_{T_0} (t-t_0))^{m+1}.
\end{eqnarray*}

Let $\epsilon>0$. Fix $t\in \left[t_0,t_0+\delta_{T_0}\right]$. Choose $m$ such that $$\dfrac{2K_{T_0}^n}{n!}(n+m)^n(4||V||_\infty K_{T_0} (t-t_0))^{m+1}\leq\dfrac{2K_{T_0}^n}{n!}(n+m)^n\left(\dfrac{1}{2}\right)^{m+1}<\dfrac{\epsilon}{3}$$ where the first inequality is valid by the choice of $\delta_{T_0}$. Then since $\lim\limits_{N\rightarrow\infty}\left|\left|E_{n,N}(t_0)\right|\right|_1=0$ by Proposition~\ref{EquivDefChaos}, we can choose $N_1 \in \N$ large enough such that
$$\left|\left|E_{n,N}(t_0)\right|\right|_1 + \sum\limits_{k=1}^m (4||V||_\infty)^k \sum\limits_{j_1=1}^n\sum\limits_{j_2=1}^{n+1}\cdots\sum\limits_{j_k=1}^{n+k-1} \dfrac{(t-t_0)^k}{k!}||E_{n+k,N}(t_0)||_1 < \dfrac{\epsilon}{3}$$ for all $N \geq N_1$.

Then choose $N_2\in\N$ such that
$$\dfrac{5n^2}{N}||V||_\infty (t-t_0)+\sum\limits_{k=1}^m (4||V||_\infty)^k \sum\limits_{j_1=1}^n\sum\limits_{j_2=1}^{n+1}\cdots\sum\limits_{j_k=1}^{n+k-1} \dfrac{5(n+k)^2}{N}||V||_\infty\dfrac{(t-t_0)^{k+1}}{(k+1)!} < \dfrac{\epsilon}{3}$$ for all $N\geq N_2$. For $N\geq \max\{N_1,N_2\}$,
$$||E_{n,N}(t)||_1<\epsilon,$$ i.e. $\lim\limits_{N\rightarrow\infty}||E_{n,N}(t)||_1=0$ for all $t\in[t_0,t_0+\delta_{T_0}]$. Since, for $n=1$,\\ $\left|\left|E_{1,N}(t)\right|\right|_1=\left|\left|\rho_N^{(1)}(t)-\rho(t)\right|\right|_1\xrightarrow[N\rightarrow\infty]{}0$ and $\rho_N^{(1)}$ are density operators for all $N$, we now obtain that $\rho(t)\in\mathcal{D}(\mathbb{H})$ for all $t\in[t_0,t_0+\delta_{T_0}]$. Therefore $\rho_N^{(n)}(t)$ is $\rho(t)$-chaotic for all $t\in[t_0,t_0+\delta_{T_0}]$.
\end{proof}

\begin{Rmk}
In the beginning of the previous proof we only knew from \cite{BoveDaPratoFano} that the solution $\rho(t)$ to the Hartree equation~(\ref{HartreeEq}) is a self-adjoint trace class operator and we set $K_{T_0}=\sup\{||\rho(t)||_1:t\in [T_0,T_0+1]\}$ for each $T_0\in\mathbb{N}\cup\{0\}$ where $\rho(t)$ is the solution to the Hartree equation~(\ref{HartreeEq}). At the end of the previous proof, we concluded that $\rho(t)$ is a positive and trace $1$ operator, hence a density operator for every $t\geq 0$. Thus $K_{T_0}=1$ for all $T_0\in\mathbb{N}\cup\{0\}$, and $\delta_{T_0}$ which is defined in the beginning of the previous proof can be picked infinitesimally close to $\dfrac{1}{4||V||_{\infty}}$ if $V\neq 0$. If $V=0$, then equation~(\ref{inequality1}) becomes $||E_{n,N}(t)||_1\leq ||E_{n,N}(0)||_1$. Therefore, the rate of chaos does not increase.
\end{Rmk}

\end{document}